\documentclass[11pt,letterpaper]{article}

\usepackage[margin=1in]{geometry}
\usepackage[utf8]{inputenc}
\usepackage[T1]{fontenc}
\usepackage{lmodern}
\usepackage{microtype}
\usepackage{colortbl}

\usepackage{amsmath,amssymb,amsthm}
\usepackage{mathtools}
\usepackage{bm}
\usepackage{graphicx}
\usepackage{booktabs}
\usepackage{algorithm}
\usepackage{algpseudocode}
\usepackage{enumitem}
\usepackage{xcolor}
\usepackage{tikz}
\usepackage[hidelinks,breaklinks=true]{hyperref}
\usepackage[capitalise,noabbrev]{cleveref}
\usepackage[square,sort&compress,numbers]{natbib}
\newcommand{\pvec}{\mathbf{p}}
\newcommand{\svec}{\mathbf{s}}
\newcommand{\avec}{\mathbf{a}}
\newcommand{\cP}{\mathcal{P}}
\newcommand{\cE}{\mathcal{E}}
\newcommand{\cD}{\mathcal{D}}

\newcommand{\cM}{\mathcal{M}}
\newcommand{\cB}{\mathcal{B}}
\newcommand{\EE}{\mathbb{E}}
\newcommand{\NN}{\mathbb{N}}
\newcommand{\RR}{\mathbb{R}}
\newcommand{\Kmax}{K_{\max}}
\newcommand{\Univ}{\mathfrak{U}}
\newcommand{\Ep}{\mathrm{Ep}}
\newcommand{\argmax}{\operatorname*{arg\,max}}

\newcommand{\Regret}{\mathrm{Regret}}
\newcommand{\Api}{\mathcal{A}^{\mathrm{pi}}}

\theoremstyle{plain}
\newtheorem{theorem}{Theorem}
\newtheorem{lemma}{Lemma}
\newtheorem{proposition}{Proposition}
\newtheorem{corollary}{Corollary}
\theoremstyle{definition}
\newtheorem{definition}{Definition}
\newtheorem{problem}{Problem}
\theoremstyle{remark}
\newtheorem{remark}{Remark}

\usepackage{fancyhdr}
\title{\textbf{\LARGE No-Regret Learning in Stackelberg Security Games\\
with an Unknown, Bounded Set of Attacker Types}}
\author{%
\begin{tabular}{c@{\hspace{4em}}c}
Sabil Ahmad$^{*}$ & Ishita$^{*}$\\
\small\texttt{23b1057@iitb.ac.in} & \small\texttt{23b0921@iitb.ac.in}\\[10pt]
Aditi Singh$^{*}$ & Swaprava Nath\\
\small\texttt{23b1053@iitb.ac.in} & \small\texttt{swaprava@iitb.ac.in}\\[10pt]
\multicolumn{2}{c}{Indian Institute of Technology Bombay}
\end{tabular}}
\date{}

\begin{document}
\maketitle
\begingroup
\renewcommand{\thefootnote}{*}
\footnotetext{Equal contributions.}
\endgroup

\begin{abstract}
In a repeated Stackelberg security game, a defender commits in every round
to a randomized allocation of its resources, and an attacker, drawn from a
set of types, best-responds. Existing no-regret guarantees assume that the
defender knows this set of types in advance. We drop this assumption: the
set of types is unknown, and its size is at most $\Kmax$, an upper bound
that the defender is not told. We show that what can be learned then
depends on what the defender observes. If it observes the type of each
attacker, an algorithm that adds types to its list as they appear achieves
regret $\tilde O(n\Kmax\sqrt{T})$ over $T$ rounds with $n$ targets, the same
$\sqrt T$ rate as with known types, at an extra factor of only
$\sqrt{\Kmax}$ and without ever using $\Kmax$. If it observes only the
attacked target and types come from an infinite universe, no algorithm has
sublinear regret: in a fixed three-target game every algorithm suffers
regret at least $T-2$, even against a single attacker type, because it
cannot tell which type it faces. If, in addition, it is told every $B$
rounds which types have appeared so far, sublinear regret returns, at
$\tilde O(n\Kmax^{4/3}T^{2/3}+\Kmax B)$. The first term is the rate of the
partial-information learner we use, and we show that the additive cost of
the delay cannot be avoided.
\end{abstract}

\section{Introduction}
\label{sec:intro}

In a Stackelberg security game (SSG), a defender allocates a limited set of
resources to protect $n$ targets against a strategic attacker who observes
the defender's randomized deployment and then attacks the target it likes
best \citep{tambe2012}. SSGs underlie several deployed systems: patrols of
the US Coast Guard in ports \citep{shieh2012}, the scheduling of Federal Air
Marshals on flights \citep{tsai2009}, checkpoints and canine patrols at Los
Angeles International Airport \citep{pita2008}, and ranger patrols against
poaching in national parks \citep{yang2014,fang2016}. In each of them the
defender commits to its deployment before any attack, and the attacker's
payoffs encode which targets it values and how much it fears being caught.

The classical analysis assumes that the defender knows the attacker's
payoffs; the optimal commitment is then the solution of a linear program
\citep{conitzer2006}. In practice the defender faces many attackers whose
payoffs it knows only roughly. \citet{balcan2015} address this through
online learning. The attacker of each round is one of a finite set $K$ of
\emph{types}, chosen by an adversary, and the defender seeks low
\emph{regret}, that is, a total payoff close to that of the best fixed
deployment in hindsight. If the defender learns the type after each round,
regret $\tilde O(n\sqrt{|K|T})$ is achievable; if it sees only which target
was attacked, $\tilde O(n|K|T^{2/3})$ is, and \citet{bernasconi2023} later
improved this to $\tilde O(\mathrm{poly}(|K|)\sqrt T)$. If the number of types is
unrestricted, regret is linear in the worst case.

These guarantees require the defender to know the set $K$ in advance, with
the payoffs of every type in it. This is rarely the case. A new group of
poachers may enter a park and target different species; a new smuggling
network may favor different routes through a port; in cyber defense, new
threat actors with new objectives appear routinely. What the defender
can reasonably assume is that only a limited number of such adversaries
are active, without knowing who they are or exactly how many. The
information the defender receives also varies across these settings. After
an incident it may identify who attacked, as when forensic evidence or
captured equipment reveals the poacher group. It may learn only which
target was hit, as when an intrusion is detected but not attributed. Or it
may receive periodic intelligence reports listing the groups known to be
active, without saying which of them carried out which attack. This paper
asks:
\begin{center}
\emph{Can the defender achieve sublinear regret when the set of attacker
types is unknown and only its size is bounded, and how does the answer
depend on what the defender observes?}
\end{center}
The question lies between the two regimes of \citet{balcan2015}: a known
set of types, where sublinear regret is possible, and an unrestricted
number of types, where it is not.

\subsection{Our contributions}
\label{sec:contrib}

We answer the question for each of the three kinds of feedback above. The
answers share one message: a bound on the \emph{number} of types helps
only if the defender can learn \emph{which} types it faces. Throughout,
$\Kmax$ is an upper bound on the number $|K|$ of types, and the defender is
told neither the set $K$ nor $\Kmax$.

\begin{itemize}[leftmargin=1.1em, labelsep=0.4em]
\item \emph{Observing the type: no-regret at the known-type rate}
(\Cref{sec:unknownK}). Our algorithm \textsc{Grow-Catalog}
(\Cref{alg:gc}) keeps a list of the types seen so far and runs the
known-type learner of \citet{balcan2015} on it, restarting the learner
whenever a new type appears. \Cref{thm:main} shows that its regret is
$O\bigl(n\Kmax\sqrt{T\log(n\Kmax)}\bigr)$. This is the first no-regret
guarantee when an adversary chooses the types from an unknown set. It has the same $\sqrt T$
dependence as the known-type bound, and the price of not knowing the types
is a factor $\sqrt{\Kmax}$. The algorithm never uses $\Kmax$, and its regret
scales with the number of types that actually appear (\Cref{cor:adaptive}).
The $\sqrt T$ rate cannot be improved: as in any expert problem, even with
two known types and full information, regret $\Omega(\sqrt T)$ is
unavoidable (\Cref{prop:sqrtT}).

\item \emph{Observing only the target: impossibility, even with one type}
(\Cref{sec:partial}). \Cref{thm:impossible} shows that when the defender
sees only the attacked target and the types come from an infinite universe,
every algorithm, randomized or not, suffers regret at least $T-2$ on a fixed
three-target game, even when a single type attacks in every round. The
linear lower bound of \citet{balcan2015} needs a number of types that grows
with $T$, even when the type is observed; ours needs one type, fixed before
the game starts. The obstacle is
therefore not how many types there are, but that the defender cannot
identify the one it faces. The types form a continuum, and each coverage
vector catches at most one of them, so the defender must find a good
deployment by exact-match search. Each round reveals at most one bit about
the type, and the counting argument of \Cref{lem:count} shows that this is
too little to win such a search in time. The result also differs from the
impossibility of \citet{bollini2026}, where the follower's type is drawn
from a distribution; here the attacker is one fixed type.

\item \emph{Periodic lists of types: sublinear regret again}
(\Cref{sec:block}). Suppose the defender sees only the attacked target, but
every $B$ rounds it is told the set of types that have attacked so far,
not when or how often. Our algorithm \textsc{Block-Grow-Catalog}
(\Cref{alg:bgc}) restores sublinear regret,
$O\bigl(n\Kmax^{4/3}T^{2/3}\log^{1/3}(n\Kmax)+\Kmax B\bigr)$
(\Cref{thm:block}). The key tool is a partial-information learner with an
anytime guarantee (\Cref{lem:anytime}), built with the doubling trick
\citep{auer1995,cesa2006}; we need it because the defender does not know how
long each epoch will last. The $T^{2/3}$ rate is that of
this learner, built on the partial-information learner of
\citet{balcan2015}, and our analysis uses it as a black box; any anytime
learner with a better rate for known types gives a correspondingly better
bound. The additive term is the cost of the delay, and it is tight:
\Cref{thm:blocklb} shows that every algorithm must pay $\Kmax(B-2)$.
Reports are needed rarely: $O(T^{1/3})$ of them over the whole horizon
leave the rate unchanged (\Cref{cor:B}).
\end{itemize}
Experiments in \Cref{sec:exp} illustrate the behavior predicted by the
analysis. \Cref{tab:landscape} places these results among the known ones.

\begin{table}[!ht]
\centering
\footnotesize
\renewcommand{\arraystretch}{1.2}
\begin{tabular}{@{}>{\raggedright\arraybackslash}p{3.0cm}p{3.8cm}>{\raggedright\arraybackslash}p{2.8cm}>{\raggedright\arraybackslash}p{5.6cm}@{}}
\toprule
Set of types & Regret & Result & What is new\\
\midrule
\multicolumn{4}{@{}l}{\emph{The defender observes the type of the attacker}}\\
known                     & $\tilde O(n\sqrt{|K|T})$ & \citet[Thm.~5.1]{balcan2015} & \\
known, $|K|$ unrestricted & $\Omega(T)$              & \citet[Thm.~7.1]{balcan2015} & \\
known, $|K|=2$            & $\Omega(\sqrt T)$        & \Cref{prop:sqrtT}            & standard expert lower bound, here in an SSG\\
\rowcolor{black!8}
\textbf{unknown}, $|K|\le\Kmax$ & $\tilde O(n\Kmax\sqrt{T})$ & \Cref{thm:main,cor:adaptive} & first bound for adversarial unknown types; $\Kmax$ not used; adapts to the types that appear\\
\midrule
\multicolumn{4}{@{}l}{\emph{The defender observes only the attacked target}}\\
known                     & $\tilde O(n|K|T^{2/3})$  & \citet[Thm.~6.1]{balcan2015} & \\
known                     & $\tilde O(\mathrm{poly}(|K|)\sqrt{T})$ & \citet{bernasconi2023} & \\
\rowcolor{black!8}
\textbf{unknown}, even $|K|=1$  & $\Omega(T)$  & \Cref{thm:impossible} & one fixed type suffices; \citet[Thm.~7.1]{balcan2015} need $|K|$ growing with $T$ even when types are observed\\
\midrule
\multicolumn{4}{@{}l}{\emph{The defender observes the attacked target, and every $B$ rounds the types seen so far}}\\
\rowcolor{black!8}
\textbf{unknown}, $|K|\le\Kmax$ & $\tilde O(n\Kmax^{4/3}T^{2/3}+\Kmax B)$ & \Cref{thm:block} & sublinear for $B=o(T)$; works with any anytime learner (\Cref{lem:anytime})\\
\rowcolor{black!8}
\textbf{unknown}, $|K|\le\Kmax$ & $\Omega(\Kmax B)$ & \Cref{thm:blocklb} & the delay cost is tight; $O(T^{1/3})$ reports suffice\\
\bottomrule
\end{tabular}
\caption{Regret in the repeated SSG, grouped by what the defender observes.
$n$: targets, $T$: rounds, $\Kmax$: upper bound on $|K|$, not known to the
defender; $\tilde O$ hides logarithmic factors; $\Omega$ bounds hold for
every algorithm. Shaded rows are the results of this paper.}
\label{tab:landscape}
\end{table}

\subsection{Related work}
\label{sec:related}

Computing the optimal commitment against a single known attacker is a
linear program \citep{conitzer2006}, and compact formulations scale it to
large deployments \citep{kiekintveld2009,korzhyk2010}. Uncertainty about the
attacker has been handled by Bayesian SSGs, which optimize against a known
prior over types \citep{paruchuri2008,kiekintveld2011}, by robust
formulations over payoff intervals \citep{pita2010,kiekintveld2013}, and by
behavioral models fitted to past attacks \citep{nguyen2013,haghtalab2016}.
Each of these assumes a known population of attackers, given as a prior, an
uncertainty set or a parametric model.

A second line learns the optimal commitment against one fixed but unknown
attacker from its best responses \citep{letchford2009,blum2014b,peng2019},
or from revealed preferences \citep{roth2016,dong2018}. The number of
queries these methods need depends on the geometry of the attacker's
best-response regions. This is consistent with our impossibility result,
whose adversary chooses a type with an exponentially thin region
(\Cref{sec:impossibility}).

Online learning against a sequence of attackers was introduced by
\citet{balcan2015}, whose setting we build on: an oblivious adversary picks
each round's type from a finite set that the defender knows, together with
every type's payoffs. They obtain $\tilde O(n\sqrt{|K|T})$ regret under full
information and $\tilde O(n|K|T^{2/3})$ under partial information.
\citet{bernasconi2023} improve the partial-information rate to
$\tilde O(\mathrm{poly}(|K|)\sqrt T)$ through a reduction to online linear
optimization, and \citet{harris2024side} and \citet{balcan2026} extend the
setting to side information observed by both players, reaching
$\tilde O(\sqrt T)$ under bandit feedback. \citet{chen2026} allow several
attackers and varying attack intensities in each round. Other work studies
attackers that are not myopic \citep{haghtalab2022}, or combines bandit
algorithms with equilibrium computation when the defender has no prior
knowledge \citep{xu2016,klima2015,marecki2012}. Every upper bound in this
line requires the set of types, and their payoffs, to be known before the
first round; our results remove that requirement.

The work closest to ours is that of \citet{bollini2026}, who study Bayesian
Stackelberg games in which the leader knows neither the follower's payoffs
nor the distribution of its types. They show that no-regret learning is
impossible when the leader observes only the follower's action, and give an
$\tilde O(\sqrt T)$ algorithm when it also observes the type.
\citet{personnat2026} study several followers whose utilities are known but
whose type distribution is not. Both papers draw types i.i.d.\ from a fixed
distribution, whereas our types form an adversarial sequence and our
benchmark is the best fixed strategy on that sequence. Our impossibility
result holds in a security game with a single type, where theirs uses three
types in a general Stackelberg game; theirs, in turn, needs only polynomial
precision in the payoffs, where ours needs precision exponential in $T$.
Neither paper considers revelation of types after a delay.

Our algorithms run an expert algorithm
\citep{littlestone1994,freund1997,cesabianchi2007,cesa2006} over a set of
experts that grows when a new type appears. Expert sets that change over
time have been studied as sleeping experts
\citep{freund1997b,kleinberg2010}, growing sets \citep{mourtada2017} and
dying experts \citep{shayestehmanesh2019}. There the learner is told when
the set changes, and the losses of existing experts do not. In our setting
a new type changes the loss of every existing expert, since an expert is a
coverage vector and its payoff depends on how the new type responds to it,
so we restart the learner instead. Under block revelation the defender
learns of a new type only after a delay, which relates our additive term to
online learning with delayed feedback
\citep{weinberger2002,joulani2013,quanrud2015}; there the delay applies to
the losses, here to the identity of the attacker.

\section{Preliminaries}
\label{sec:model}

A defender protects $n$ targets with limited resources against a sequence of
attackers over $T$ rounds. In each round the defender commits to a randomized
allocation of its resources; the attacker of that round sees the
allocation, but not its random outcome, and attacks the target it likes
best; both players then receive payoffs that depend on whether the attacked
target was covered. Attackers differ in their payoffs, and each attacker is
one of finitely many \emph{types}. The first four paragraphs below define
this model, and the last recalls how the defender learns when the set of
types is known.

\paragraph{Game primitives.}
There are $n$ targets $N=\{1,\dots,n\}$. The defender has $R$ resources and a
collection $\cD\subseteq 2^N$ of \emph{schedules}, each a set of
targets that one resource can protect at the same time. An
\emph{assignment function} $A:R\to 2^{\cD}$ lists the schedules
each resource may take. A \emph{pure strategy} assigns every resource to one
of its allowed schedules; let $m$ be the number of pure strategies. A
\emph{mixed strategy} is a distribution $\svec\in\Delta_m$ over pure
strategies. Its \emph{coverage vector} is $\pvec=M\svec\in[0,1]^n$, where
$M\in\{0,1\}^{n\times m}$ has $M_{ij}=1$ if and only if the $j$-th pure
strategy covers target $i$, so $p_i$ is the probability that target $i$ is
covered. Let $\cP$ be the set of coverage vectors of all mixed strategies.
It is the convex hull of the coverage vectors of the pure strategies, and
hence a convex polytope. For example, with $n=3$ targets and one guard who
can stand at any single target, the pure strategies are the three targets,
and $\cP=\{\pvec\in[0,1]^3:p_1+p_2+p_3=1\}$. Payoffs and the attacker's
behavior depend on a mixed strategy only through its coverage vector, so
from now on the defender is taken to choose $\pvec\in\cP$ directly.

\paragraph{Defender payoffs.}
For each target $i$, the defender receives $u_d^c(i)\in[-1,1]$ if $i$ is
attacked and covered, and $u_d^u(i)\in[-1,1]$ if $i$ is attacked and
uncovered. Averaging over the defender's own randomization, under which $i$
is covered with probability $p_i$, the defender's utility when target $i$ is
attacked is
\begin{equation}
U_d(i,\pvec)\;=\;u_d^c(i)\,p_i \;+\; u_d^u(i)\,(1-p_i),
\label{eq:Ud}
\end{equation}
which is linear in $\pvec$ for each fixed target $i$.

\paragraph{Attacker types and best response.}
An attacker \emph{type} $\alpha$ is a pair
$(u_\alpha^c,u_\alpha^u)\in[-1,1]^n\times[-1,1]^n$ of per-target payoffs,
for a covered and an uncovered attack respectively. Its utility for
attacking target $i$ is $U_\alpha(i,\pvec)=u_\alpha^c(i)\,p_i+u_\alpha^u(i)\,(1-p_i)$.
Seeing $\pvec$, type $\alpha$ attacks a target in
$\argmax_{i\in N}U_\alpha(i,\pvec)$, breaking ties by a fixed order over the
targets, which may depend on the type; we write $b_\alpha(\pvec)$ for the
target it attacks. The map $b_\alpha$ is piecewise constant, so the
defender's payoff $U_d(b_\alpha(\pvec),\pvec)$ jumps where the attacker
switches targets. For example, with one guard and two targets, an attacker
who values both targets equally attacks whichever is less covered, and the
defender's payoff jumps as $p_1$ crosses $1/2$. These jumps are why standard
methods for online linear optimization do not apply directly, and the last
paragraph of this section explains how they are handled.

\paragraph{Repeated game and regret.}
There is a set $K$ of attacker types; throughout, $|K|$ denotes its size.
An \emph{oblivious} adversary fixes a sequence
$\avec=(a_1,\dots,a_T)\in K^T$ before the game starts, and the defender does
not know it. In round $t=1,\dots,T$ the defender commits to
$\pvec_t\in\cP$, the attacker $a_t$ attacks $b_{a_t}(\pvec_t)$, and the
defender receives $U_d(b_{a_t}(\pvec_t),\pvec_t)$. Under
\emph{full-information feedback}, at the end of round $t$ the defender
learns the type $a_t$, and the first time a type is seen its payoffs
$(u_\alpha^c,u_\alpha^u)$ are revealed as well. The benchmark is the best
fixed coverage vector in hindsight,\footnote{Because of tie-breaking the
maximum in \Cref{eq:pstar} may not be attained; as in
\citet[footnote~1]{balcan2015}, $\pvec^{*}$ then denotes a coverage vector
whose total payoff is arbitrarily close to the supremum.}
\begin{equation}
\pvec^{*}\;=\;\argmax_{\pvec\in\cP}\;\sum_{t=1}^T U_d\bigl(b_{a_t}(\pvec),\pvec\bigr),
\label{eq:pstar}
\end{equation}
and the regret of a defender that plays $\pvec_1,\dots,\pvec_T$ is
\begin{equation}
\Regret(T)\;=\;\sum_{t=1}^T\!U_d(b_{a_t}(\pvec^{*}),\pvec^{*})\;-\;\EE\!\bigg[\sum_{t=1}^T\!U_d(b_{a_t}(\pvec_t),\pvec_t)\bigg],
\label{eq:regret}
\end{equation}
where the expectation is over the defender's internal randomization. Up to
this point the model is that of \citet[Sec.~2]{balcan2015}, with the same
notation; the difference, studied in the rest of the paper, is that the
defender does not know $K$.

\paragraph{Learning with a known type set.}
\label{sec:polytopes}
Suppose for the moment that the defender knows a finite set $\Theta$ of types
containing every attacker. \citet{balcan2015} solve this case, and our
algorithms use their solution as a subroutine, with $\Theta$ the set of types
seen so far. It combines two ingredients.

\emph{An expert algorithm.} An \emph{expert} is a fixed action that the
learner may choose; here it will be a fixed coverage vector. An
\emph{expert algorithm} keeps a weight on each of finitely many experts,
plays an expert drawn with probability proportional to its weight, then
sees the loss of every expert and lowers the weights of those that did
badly. Polynomial Weights \citep{cesabianchi2007}, for instance, multiplies
each weight by $1-\eta\cdot\text{loss}$ for a learning rate $\eta$.

\begin{proposition}[{\citet[Proposition~1]{balcan2015}, after \citet{blum2007}; anytime form by the doubling trick \citep[Sec.~2.3]{cesa2006}}]
\label{prop:hedge}
Consider a learner that in each round chooses a distribution over a finite
set $\cM$ of experts, after which every expert receives a loss in
$[-\kappa,\kappa]$. There is an expert algorithm whose expected regret after
$T'$ rounds, against the best fixed expert in hindsight, is
$O\bigl(\sqrt{T'\kappa\log|\cM|}\bigr)$ for every $T'$, and which is not told
$T'$.
\end{proposition}

\emph{A finite set of experts.} \Cref{prop:hedge} cannot use every coverage
vector as an expert, since $\cP$ is infinite. For a \emph{response profile}
$\sigma:\Theta\to N$, which names a target for every type, let
\begin{equation}
\cP_\sigma(\Theta)\;=\;\{\pvec\in\cP:\ b_\alpha(\pvec)=\sigma(\alpha)\text{ for every }\alpha\in\Theta\}
\label{eq:Psigma}
\end{equation}
be the region where every type attacks the target $\sigma$ names. It is cut
out of $\cP$ by linear inequalities, so it is a convex polytope, and the
nonempty regions partition $\cP$. Inside one region every attacker's
target is fixed, so by \Cref{eq:Ud} the defender's total payoff is linear in
$\pvec$ and is maximized at a corner. The corners of all regions would
therefore be enough experts, except that tie-breaking makes some defining
inequalities strict: a corner may lie just outside its region, where some
attacker already switches target. Each such corner is replaced by a point of
its region within $\ell_1$-distance $\epsilon$ of it. The resulting finite
set $\cE(\Theta;\epsilon)$ of \emph{$\epsilon$-approximate extreme points}
loses little and is small.

\begin{lemma}[{\citet[Lemma 4.3]{balcan2015}}]
\label{lem:43}
For every finite $\Theta$, every $\epsilon>0$ and every sequence
$(a'_1,\dots,a'_{T'})\in\Theta^{T'}$,
\[
\max_{\pvec\in\cE(\Theta;\epsilon)}\!\sum_{t=1}^{T'}\!U_d\bigl(b_{a'_t}(\pvec),\pvec\bigr)\;\ge\;\sup_{\pvec\in\cP}\!\sum_{t=1}^{T'}\!U_d\bigl(b_{a'_t}(\pvec),\pvec\bigr)-2\epsilon\,T'.
\]
\end{lemma}

\begin{lemma}[{\citet[Lemma 4.4]{balcan2015}}]
\label{lem:44}
For every finite $\Theta$ and every $\epsilon>0$,
$\log|\cE(\Theta;\epsilon)|=O\bigl(n^2|\Theta|\log(n|\Theta|)\bigr)$, a
bound that does not depend on $\epsilon$.
\end{lemma}

\emph{Putting them together.} Under full information the defender knows
$a_t$ and the payoffs of every type in $\Theta$, so it can compute the loss
$\ell_t(\pvec)=-U_d(b_{a_t}(\pvec),\pvec)$ of every expert in
$\cE(\Theta;\epsilon)$. Running an expert algorithm on
$\cE(\Theta;\epsilon)$ with these losses is \citet[Theorem~5.1]{balcan2015}.
We use it with two changes, both needed later when the learner is
restarted at times the defender cannot foresee. First, since the size of
$\cE(\Theta;\epsilon)$ does not depend on $\epsilon$ (\Cref{lem:44}),
$\epsilon$ can be made as small as we like at no cost; we fix $\epsilon=1/T$
from the global horizon, so that $L$ rounds lose at most $2L/T$ to the
approximation, and any collection of disjoint stretches of rounds loses at
most $2$ in total. Second, by \Cref{prop:hedge} the bound holds after every
number of rounds $L$, without $L$ being supplied.

\begin{lemma}[{Known type set; adapted from \citet[Theorem~5.1]{balcan2015}}]
\label{lem:known}
Let $\Theta$ be a finite set of types, $\epsilon=1/T$, and
$\mathcal{A}(\Theta)$ an algorithm satisfying \Cref{prop:hedge}, run on the
experts $\cE(\Theta;\epsilon)$ with the losses $\ell_t$. There is an
absolute constant $c_1$ such that, for every $L\in\{1,\dots,T\}$ and every
sequence $(a_1,\dots,a_L)\in\Theta^L$ fixed in advance,
\begin{equation}
\sup_{\pvec\in\cP}\sum_{t=1}^{L}U_d\bigl(b_{a_t}(\pvec),\pvec\bigr)
-\EE\Bigl[\sum_{t=1}^{L}U_d\bigl(b_{a_t}(\pvec_t),\pvec_t\bigr)\Bigr]
\;\le\;c_1\,n\sqrt{|\Theta|\,L\,\log\bigl(n|\Theta|\bigr)}\;+\;\frac{2L}{T}.
\label{eq:known}
\end{equation}
The learner $\mathcal{A}(\Theta)$ is not told $L$.
\end{lemma}

\begin{proof}
The left side of \Cref{eq:known} splits into two gaps. The first is between
the best expert and $\mathcal{A}(\Theta)$. Losses lie in $[-1,1]$, so
\Cref{prop:hedge} with $\kappa=1$ bounds it by
$O\bigl(\sqrt{L\log|\cE(\Theta;\epsilon)|}\bigr)$, which by \Cref{lem:44} is
$O\bigl(n\sqrt{|\Theta|L\log(n|\Theta|)}\bigr)$. The second is between the
supremum over $\cP$ and the best expert, at most $2\epsilon L=2L/T$ by
\Cref{lem:43}. Adding the two gives \Cref{eq:known}.
\end{proof}

\Cref{lem:known} needs every attacker to lie in $\Theta$. When $K$ is
unknown, the defender still knows the types it has seen, and until a new
type appears every attacker is among them. The next section builds an
algorithm on exactly this observation.

\section{Unknown Type Set and Full Information: Sublinear Regret}
\label{sec:unknownK}

This section treats the case in which the set of attacker types is unknown
but the defender learns, after each round, which type attacked.

\begin{problem}[Unknown-$K$ repeated SSG]\label{prob:unknownK}
A repeated SSG as in \Cref{sec:model} is played under the following
conditions.
\begin{enumerate}[leftmargin=1.2em,itemsep=0.05em,topsep=0.2em]
  \item The set $K$ of attacker types is finite, and $|K|\le\Kmax$ for some
  constant $\Kmax\ge1$. Neither $K$ nor $\Kmax$ is revealed to the defender.
  \item Feedback is full information: after round $t$ the defender
  observes $a_t$, and the first time any type $\alpha$ appears, its
  utility vectors $(u_\alpha^c,u_\alpha^u)$ are revealed.
\end{enumerate}
The adversary is oblivious, and the benchmark $\pvec^{*}$ and regret
$\Regret(T)$ are as in \Cref{eq:pstar,eq:regret}.
\end{problem}

\begin{remark}[a bound on $|K|$ is necessary]\label{rem:necessity}
\citet[Theorem 7.1]{balcan2015} exhibit, for every $T$, a game with
$|K|<2^{T+1}$ on which every online algorithm has regret at least $T/2$.
Hence sublinear regret needs some bound on $|K|$ that does not grow with
$T$. \Cref{prob:unknownK} asks only that such a bound exist, not that the
defender know it.
\end{remark}

\paragraph{The \textsc{Grow-Catalog} algorithm.}
\label{sec:alg}
The defender does not know $K$, but at every round it knows which types it
has observed so far, together with their utilities. We call this set the
\emph{catalog} and write $C_t\subseteq K$ for the catalog at round $t$. It is
the only information about the attacker population on which the choice of
$\pvec_t$ can depend, and it is needed because the defender can optimize only
against types it knows. \Cref{alg:gc} treats the catalog as if it were the
full type set. The rounds at which a type appears for the first time cut the
horizon into \emph{epochs}. Within an epoch the catalog does not change and
every attacker that appears belongs to it, so the defender faces a game with a
known type set $\Theta=C_t$ and runs the learner of \Cref{lem:known} on it.
When a new type appears, the catalog grows, the learner is restarted on the
enlarged catalog, and the next epoch begins.

Formally, $\mathcal{A}(C)$ denotes the learner of \Cref{lem:known} with
$\Theta=C$: an expert algorithm satisfying \Cref{prop:hedge}, for
instance Polynomial Weights, run over the expert set
$\cE(C;\epsilon)$ with $\epsilon=1/T$ and the losses
$\ell_t(\pvec)=-U_d(b_{a_t}(\pvec),\pvec)$. Full-information feedback reveals
$a_t$ and the utilities of every type in the catalog, so these losses can be
evaluated at every expert. The accuracy $\epsilon$ is fixed once from the
horizon $T$, which the defender knows, and $\mathcal{A}(C)$ does not need the
length of an epoch in advance; the proof of \Cref{thm:main} uses both facts.
The round of first contact with a new type falls outside that guarantee and
is charged separately in the proof. \Cref{alg:gc} never uses $\Kmax$.

\begin{algorithm}[t]
\caption{\textsc{Grow-Catalog} (full information).}
\label{alg:gc}
\begin{algorithmic}[1]
\small
\Require horizon $T$,
   game primitives, default $\pvec_{\text{def}}\in\cP$
\State $\epsilon\gets 1/T$
   \Comment{fixed once, from the global horizon}
\State $C\gets\emptyset$;~$\mathcal{A}\gets\textsc{null}$
   \Comment{empty catalog, no inner learner yet}
\For{$t=1,\dots,T$}
  \State \textbf{Play} $\pvec_t\gets\pvec_{\text{def}}$ if $C=\emptyset$, else the strategy $\mathcal{A}$ prescribes for round $t$
  \State \textbf{Observe} the attacking type $a_t$
  \If{$a_t\notin C$} \Comment{discovery round: new attacker type}
    \State $C\gets C\cup\{a_t\}$; record $(u_{a_t}^c,u_{a_t}^u)$
    \State recompute $\cE(C;\epsilon)$
       \Comment{the enlarged expert set, \Cref{lem:43}}
    \State $\mathcal{A}\gets\mathcal{A}(C)$
       \Comment{restart the learner of \Cref{lem:known} on $\cE(C;\epsilon)$}
  \Else \Comment{routine round}
    \State pass $a_t$ to $\mathcal{A}$ as its round-$t$ feedback
  \EndIf
\EndFor
\end{algorithmic}
\end{algorithm}

Our main result for full information is that this simple scheme keeps the
$\sqrt T$ rate of the known-type case.

\begin{theorem}[unknown $K$ with full information is learnable]
\label{thm:main}
Let the defender play \Cref{alg:gc} in \Cref{prob:unknownK}. For any
oblivious sequence $\avec\in K^T$,
\begin{equation}
\Regret(T)\;\le\;O\!\Bigl(n\Kmax\sqrt{T\log (n\Kmax)}\Bigr).
\label{eq:main}
\end{equation}
In particular, if $\Kmax=O(1)$ then $\Regret(T)=\tilde O(n\sqrt T)$,
matching the known-type bound of \Cref{lem:known}
\citep[Theorem 5.1]{balcan2015} up to logarithmic factors.
\end{theorem}

\begin{proof}
\emph{Idea.} Cut the horizon at the rounds where a type attacks for the
first time. Between two cuts the catalog is fixed and contains every
attacker, so the defender faces the known-type problem and \Cref{lem:known}
bounds its regret. Each cut round is charged the largest possible loss, and
the Cauchy--Schwarz inequality adds the per-epoch bounds.

Fix an oblivious sequence $\avec=(a_1,\dots,a_T)\in K^T$, and let
$K'=\{a_1,\dots,a_T\}\subseteq K$ be the set of types that appear in it, so
that $|K'|\le|K|\le\Kmax$. For a set of rounds $S\subseteq[T]$ write
\[
R(S)\;=\;\sum_{t\in S}U_d(b_{a_t}(\pvec^{*}),\pvec^{*})
-\EE\Bigl[\sum_{t\in S}U_d(b_{a_t}(\pvec_t),\pvec_t)\Bigr]
\]
for the regret accumulated on $S$ against the benchmark $\pvec^{*}$ of
\Cref{eq:pstar}. Then $R([T])=\Regret(T)$, and $R$ adds up over disjoint
sets of rounds.

\emph{Epochs.} Let $\tau_1<\dots<\tau_{|K'|}$ be the \emph{discovery rounds},
the rounds $t$ with $a_t\notin\{a_1,\dots,a_{t-1}\}$, and set
$\tau_{|K'|+1}=T+1$. The $i$-th \emph{epoch} is the stretch strictly between
two consecutive discovery rounds,
\begin{equation}
\Ep_i=\{\tau_i+1,\dots,\tau_{i+1}-1\},\qquad L_i=|\Ep_i|,\qquad i=1,\dots,|K'|.
\label{eq:epoch-def}
\end{equation}
Every round is either a discovery round or in exactly one epoch, so with
$\cD=\{\tau_1,\dots,\tau_{|K'|}\}$,
\begin{equation*}
\Regret(T)\;=\;R(\cD)+\sum_{i=1}^{|K'|}R(\Ep_i),
\qquad \sum_{i=1}^{|K'|}L_i\le T.
\end{equation*}
Since $\avec$ is oblivious, the discovery rounds, and hence the epochs, are
fixed before play starts and do not depend on what the defender does.

\emph{Each epoch.} Throughout $\Ep_i$ the catalog is
$C_i=\{a_{\tau_1},\dots,a_{\tau_i}\}$, every attacker of the epoch belongs to
$C_i$, and \Cref{alg:gc} runs one copy of $\mathcal{A}(C_i)$, started at the
first round of the epoch. The attackers of the epoch form a sequence in $C_i$
fixed in advance, so \Cref{lem:known} with $\Theta=C_i$ and $L=L_i$ applies.
It bounds the regret against the best coverage vector \emph{for this epoch
alone}. The global benchmark $\pvec^{*}$ is one of the candidates in that
supremum, so its payoff on the epoch is no larger, and the same bound holds
against $\pvec^{*}$:
\begin{equation}
R(\Ep_i)\;\le\;c_1\,n\sqrt{|C_i|\,L_i\log(n|C_i|)}+\frac{2L_i}{T}
\;\le\;c_1\,n\sqrt{|K'|\log(n|K'|)}\,\sqrt{L_i}+\frac{2L_i}{T}.
\label{eq:combined-epoch}
\end{equation}
The second inequality uses $|C_i|\le |K'|$ and that $x\mapsto x\log(nx)$ is
increasing. The learner does not need $L_i$, which is unknown when the epoch
starts.

\emph{Discovery rounds.} Payoffs lie in $[-1,1]$, so a single round
contributes at most $2$ to $R$, whatever the defender plays. There are $|K'|$
discovery rounds, hence
\begin{equation}
R(\cD)\;\le\;2|K'|.
\label{eq:disc-sum}
\end{equation}

\emph{Adding up.} Summing \Cref{eq:combined-epoch} over the epochs, using
$\sum_i 2L_i/T\le 2$, and adding \Cref{eq:disc-sum} gives
$\Regret(T)\le 2|K'|+2+c_1 n\sqrt{|K'|\log(n|K'|)}\,\sum_i\sqrt{L_i}$. The
Cauchy--Schwarz inequality, applied to $(1,\dots,1)$ and
$(\sqrt{L_1},\dots,\sqrt{L_{|K'|}})$, bounds the last sum:
\begin{equation}
\sum_{i=1}^{|K'|}\sqrt{L_i}\;\le\;\sqrt{|K'|\textstyle\sum_{i}L_i}\;\le\;\sqrt{|K'|\,T}.
\label{eq:CS}
\end{equation}
Therefore
\begin{equation}
\Regret(T)\;\le\;c_1\,n|K'|\sqrt{T\log(n|K'|)}\;+\;2|K'|+2,
\label{eq:main-Kprime}
\end{equation}
and $|K'|\le\Kmax$ gives \Cref{eq:main}. Applied once to the whole horizon
with the type set known, \Cref{lem:known} would give
$c_1 n\sqrt{|K'|\,T\log(n|K'|)}$; the extra factor $\sqrt{|K'|}$ comes from
\Cref{eq:CS} and is the price of not knowing the types in advance.
\end{proof}

Since \Cref{alg:gc} never reads $\Kmax$, the proof gives more than
\Cref{thm:main}: the bound \Cref{eq:main-Kprime} depends only on the types
that actually appear.

\begin{corollary}[instance-adaptive bound]
\label{cor:adaptive}
Let $K'=\{a_1,\dots,a_T\}$ be the set of attacker types that appear in the
sequence. Then \Cref{alg:gc} achieves
\[
\Regret(T)\;\le\;O\!\Bigl(n|K'|\sqrt{T\log (n|K'|)}\Bigr).
\]
\end{corollary}

\begin{proof}
This is \Cref{eq:main-Kprime}, the bound reached in the proof of
\Cref{thm:main} before $|K'|$ is replaced by $\Kmax$.
\end{proof}

The dependence on $T$ in \Cref{thm:main} is the best possible:
\Cref{prop:sqrtT} in \Cref{sec:sqrtT-main} shows that two types already
force regret of order $\sqrt T$, even when both are known in advance. What
the defender needed in this section, however, was to see the type of each
attacker. The next section shows what happens without it.

\section{Unknown Type Set and Partial Information: An Impossibility Result}
\label{sec:partial}

\Cref{thm:main} shows that not knowing $K$ is affordable: under
full-information feedback it costs at most a factor $\sqrt{\Kmax}$ over the
known-type rate of \citet[Theorem 5.1]{balcan2015}, as the last paragraph of
its proof records. Not observing $a_t$ is affordable too: with $K$ known
and only the attacked target observed, \citet[Theorem 6.1]{balcan2015}
achieve $O(T^{2/3}n|K|\log^{1/3}(n|K|))$, and \citet{bernasconi2023} improve
this to $\tilde O(\mathrm{poly}(|K|)\sqrt T)$.

This section shows that the two relaxations cannot be combined. Under
partial information with an unknown type set, no algorithm achieves regret
$o(T)$, \emph{even when $\Kmax=1$}. Here $\Kmax=1$ means that a single type
attacks in every round; the defender knows that this type comes from a given
infinite universe of types, but not which one it is. The obstruction is
exactly this: the defender sees only a target index, and infinitely many
types produce the same index, so it can never learn \emph{which} type it
faces. What the defender lacks is the identity of its opponent, and
\Cref{sec:block} studies the cheapest way of supplying it.

\begin{problem}[Unknown-$K$ repeated SSG, partial information]
\label{prob:partial}
A repeated SSG as in \Cref{sec:model} is played, with three modifications.
\begin{enumerate}[leftmargin=1.2em,itemsep=0.05em,topsep=0.2em]
  \item There is a set $\Univ$ of attacker types, called the \emph{universe},
  known to the defender and not assumed finite.
  \item The set $K\subseteq\Univ$ is finite, and $|K|\le\Kmax$ for some
  constant $\Kmax\ge1$. Neither $K$ nor $\Kmax$ is revealed to the defender.
  \item Feedback is \emph{partial information} \citep[Section 3]{balcan2015}:
  at the end of round $t$ the defender observes only the attacked target
  $b_{a_t}(\pvec_t)$. It never observes $a_t$, and utility vectors are never
  revealed.
\end{enumerate}
The adversary is oblivious, and $\pvec^{*}$ and $\Regret(T)$ are as in
\Cref{eq:pstar,eq:regret}.
\end{problem}

The defender's own payoff is not extra feedback: it knows $u_d^c$ and
$u_d^u$, so its payoff $U_d(b_{a_t}(\pvec_t),\pvec_t)$ follows from the
observed target and its own play $\pvec_t$. Moreover, \Cref{prob:partial}
is at least as hard as each setting it weakens: revealing $a_t$ recovers
\Cref{prob:unknownK}, and revealing $K$ before the first round recovers
\citet[Section 6]{balcan2015}. Our lower bounds hold even for a defender
that is told $\Kmax$.

\subsection{The lower-bound game \texorpdfstring{$G_3$}{G3}}
\label{sec:gadget}

All of our lower bounds are proved in a single small game, which we call
$G_3$; it is the game \citet[Theorem 7.1]{balcan2015} use for their
unrestricted-$K$ lower bound. There are three targets, $N=\{1,2,3\}$, and the
defender has one resource that may guard any one of them. The schedules are
therefore $\cD=\{\emptyset,\{1\},\{2\},\{3\}\}$: guard nothing, or
guard a single target. The induced coverage vector $\pvec=(p_1,p_2,p_3)$
records the probability with which each target is guarded, and the set of
feasible coverage vectors is
\begin{equation}
\cP=\{\pvec\in[0,1]^3\;:\;p_1+p_2+p_3\le 1\}.
\label{eq:g3-polytope}
\end{equation}
Every coverage vector constructed in the proofs satisfies
$p_1+p_2+p_3=1$, so the arguments apply verbatim when $\cP$ is restricted to
the simplex.

\paragraph{The defender's payoffs.} In $G_3$ the defender loses one unit
whenever target~1 or target~3 is attacked, loses nothing when target~2 is
attacked, and in each case the loss is the same whether or not the attacked
target happened to be guarded:
\begin{equation}
u_d^c(i)=u_d^u(i)=-1\ \ (i\in\{1,3\}),\qquad u_d^c(2)=u_d^u(2)=0.
\label{eq:g3-def}
\end{equation}
Because the two payoffs coincide for each target, the defender's expected
utility \Cref{eq:Ud} does not depend on the coverage at all:
\begin{equation}
U_d(1,\pvec)=U_d(3,\pvec)=-1,\qquad U_d(2,\pvec)=0
\qquad\text{for every }\pvec\in\cP.
\label{eq:g3-payoff}
\end{equation}
Coverage has not become irrelevant; it has become purely \emph{indirect}.
The attacker's utilities do depend on $\pvec$, so by choosing coverage the
defender changes which target the attacker prefers, and that choice is the
defender's entire problem: it must arrange for the attacker to prefer
target~2. Such payoffs are allowed by the model of \Cref{sec:model}, which
requires only that all utilities lie in $[-1,1]$, never that guarding a
target reduce the damage of an attack. A round is thus worth $0$ if target~2
is attacked and $-1$ otherwise. We say the defender \emph{steers
successfully} at round $t$ if $b_{a_t}(\pvec_t)=2$ and \emph{fails}
otherwise, so that the regret against a benchmark which always steers
successfully is exactly the expected number of failures. Ties in the
attacker's best response are broken toward the lowest-indexed target,
following \citet[Section 7]{balcan2015}; target~2 is the strict maximizer at
every benchmark we exhibit, so none of our bounds depends on this convention.

\paragraph{The attacker types.} Each attacker type we use in $G_3$ is
described by a pair of thresholds $0\le r_1<r_2\le 1$. The type
$\alpha_{r_1,r_2}$ has utilities
\begin{equation}
u^c_{\alpha}=(r_1-1,\;0,\;-r_2),\qquad
u^u_{\alpha}=(r_1,\;0,\;1-r_2);
\label{eq:type-utils}
\end{equation}
that is, attacking target~1 pays the attacker $r_1-1$ if the target is
covered and $r_1$ if it is not, target~2 pays $0$ either way, and target~3
pays $-r_2$ or $1-r_2$. All six entries lie in $[-1,1]$. Substituting them
into $U_\alpha(i,\pvec)=u^c_\alpha(i)\,p_i+u^u_\alpha(i)\,(1-p_i)$ gives
\begin{equation}
U_\alpha(1,\pvec)=r_1-p_1,\quad U_\alpha(2,\pvec)=0,\quad
U_\alpha(3,\pvec)=(1-r_2)-p_3 .
\label{eq:type-U}
\end{equation}
In words: attacking target~1 is attractive when target~1 is lightly
guarded, attacking target~3 is attractive when target~3 is lightly guarded,
and target~2 is a neutral fallback worth exactly zero.

\paragraph{A one-dimensional picture.}
It is convenient to look at coverage vectors that split the single
resource between targets~1 and~3 and never guard target~2. Write these as
$\pvec=(p,0,1-p)$ for a scalar $p\in[0,1]$. Substituting $p_1=p$ and
$p_3=1-p$ into \Cref{eq:type-U} turns the three utilities into $r_1-p$, $0$
and $p-r_2$. Comparing them, and breaking ties toward the lowest index,
gives
\begin{equation}
b_{\alpha_{r_1,r_2}}\bigl((p,0,1-p)\bigr)=
\begin{cases}
1 & p\in[0,r_1],\\
2 & p\in(r_1,r_2],\\
3 & p\in(r_2,1].
\end{cases}
\label{eq:type-line}
\end{equation}
So along this segment the type $\alpha_{r_1,r_2}$ may be pictured as the
interval $(r_1,r_2]$: it strikes target~2 precisely when $p$ falls inside
the interval, and strikes target~1 or target~3 according to whether $p$
falls short of the interval or overshoots it. For example, the type
$\alpha_{0.3,0.5}$ attacks target~1 under $(0.2,0,0.8)$, target~2 under
$(0.4,0,0.6)$ and target~3 under $(0.7,0,0.3)$. Locating an unknown type
means locating its interval.

An online algorithm need not restrict its plays to this segment: it may also
guard target~2 or leave some probability unassigned. Since the lower bounds
must apply to all possible plays, the following lemma describes the
attacker's behavior over the whole of $\cP$; its proof is in
\Cref{app:proof-hit}.

\begin{lemma}[when a threshold type strikes target 2]
\label{lem:hit}
Let $0\le r_1<r_2\le 1$ and $\alpha=\alpha_{r_1,r_2}$. For every
$\pvec=(p_1,p_2,p_3)\in\cP$,
\begin{equation}
b_\alpha(\pvec)=2
\iff
p_1>r_1\ \text{ and }\ p_3\ge 1-r_2 .
\label{eq:hit-iff}
\end{equation}
Moreover, whenever $b_\alpha(\pvec)=2$, the coverage of target~1 satisfies
$p_1\in(r_1,r_2]$.
\end{lemma}

The second assertion is the key to the lower bounds. Although the
condition on $p_3$ may hold for several types at once, the single-resource
constraint forces $p_1\le r_2$, so the type $\alpha_{r_1,r_2}$ can strike
target~2 only when $p_1\in(r_1,r_2]$. It is therefore the coordinate $p_1$
that separates the types.

\subsection{A family of types with disjoint intervals}
\label{sec:family}

The lower bounds use families of threshold types whose intervals do not
overlap. The basic one chops the unit interval into $M\ge2$ equal pieces:
\begin{equation}
I_j=\Bigl(\tfrac{j-1}{M},\tfrac{j}{M}\Bigr],\qquad
\alpha^{(j)}=\alpha_{\frac{j-1}{M},\frac{j}{M}},\qquad j\in[M],
\label{eq:dyadic}
\end{equation}
and we write $\Univ_M=\{\alpha^{(1)},\dots,\alpha^{(M)}\}$. By
\Cref{eq:type-line} these $M$ types respond differently to coverage vectors
on the segment, so they are distinct. By \Cref{lem:hit}, disjoint intervals
mean that no coverage vector can serve two of them at once.

\begin{corollary}[one coverage vector steers at most one type]
\label{cor:onehit}
Let $\alpha^{(1)},\dots,\alpha^{(M)}$ be threshold types in $G_3$ whose
intervals are pairwise disjoint, for instance those of \Cref{eq:dyadic}.
For every $\pvec\in\cP$, at most one of them strikes target~2 under $\pvec$.
\end{corollary}

The proof, two lines from \Cref{lem:hit}, is in \Cref{app:proof-hit}.

\subsection{Identification is a comparison search}
\label{sec:comparison}

Suppose the true instance is $K=\{\alpha^{(j)}\}$ for one unknown index
$j\in[M]$, so that the same type attacks in every round. Then the defender's
situation is a search problem. In each round it chooses a coverage vector
and observes the attacked target, which, for coverage vectors on the segment
$(p,0,1-p)$, tells it by \Cref{eq:type-line} whether $p$ fell below, inside,
or above the unknown interval $I_j$. Two of the three possible answers,
attacks on targets~1 or~3, merely narrow the range of candidates, like a
comparison in binary search; the third, an attack on target~2, identifies
$j$ exactly, and by \Cref{cor:onehit} it is the only answer that can do so.
The defender pays one unit of regret for every round in which it has not yet
received that third answer.

The next lemma bounds how quickly any such search can succeed. It sharpens
the classical fact that a worst-case comparison search over $M$ candidates
needs $\Omega(\log_2 M)$ queries in two respects: it charges the defender
for every unsuccessful round rather than counting rounds until success, and
it controls the whole distribution of the success time. In the lemma the
attacker's type is drawn uniformly at random. This is a choice of hard
instance, not an assumption about the adversary: a lower bound need only
exhibit instances on which the defender does badly. If the regret averaged
over the draw is large, then some fixed type makes it at least as large,
and that fixed type is the instance used in \Cref{thm:impossible}. Drawing
the type at random also handles randomized defenders, since their coin
tosses are independent of the draw; this is the easy direction of Yao's
minimax principle \citep{yao1977}.

\begin{lemma}[a search over disjoint intervals cannot succeed quickly]
\label{lem:count}
Let $\alpha^{(1)},\dots,\alpha^{(M)}$ be threshold types in $G_3$ whose
intervals are pairwise disjoint. Draw $j$ uniformly at random from $[M]$, and
suppose $\alpha^{(j)}$ attacks in every round. Consider a defender whose
coverage vector $\pvec_t$ in round $t$ is determined by the targets
$i_1,\dots,i_{t-1}$ it saw attacked in the earlier rounds, where
$i_s=b_{\alpha^{(j)}}(\pvec_s)$. Let $\tau$ be the first round in which the
attacker strikes target~2, with $\tau=\infty$ if this never happens. Then
for every $u\ge 1$ and every $T\ge 1$,
\begin{equation}
\Pr[\tau\le u]\;\le\;\frac{2^{u}-1}{M},
\label{eq:tau-tail}
\end{equation}
\begin{equation}
\EE\bigl[\min\{\tau-1,\,T\}\bigr]\;\ge\;T-\frac{2^{T+1}-T-2}{M}.
\label{eq:count-final}
\end{equation}
\end{lemma}

The defender never sees $j$. It sees only its own coverage vectors and the
attacked targets, and $j$ enters solely through the attacker's best response
$b_{\alpha^{(j)}}(\pvec_s)$. The intuition is binary search in which only an
exact hit counts. Each failed round tells the defender one bit, whether
target~1 or target~3 was struck, so its possible plays form a binary tree
with at most $2^{t-1}$ branches in round $t$, and each branch can catch at
most one candidate. With $M=2^T$ candidates, half of them can be caught only
in the last round, a quarter only in the round before, and so on. This is
the mechanism behind the impossibility of online learning of thresholds,
whose Littlestone dimension is infinite \citep{littlestone1988,bendavid2009};
\Cref{lem:count} adapts it to $G_3$, where a coverage vector yields one
comparison and success needs an exact hit.

\begin{proof}[Proof sketch]
Before its first success the defender has seen only attacks on targets~1
or~3, so at the start of round $t$ its history is one of at most $2^{t-1}$
strings over $\{1,3\}$, and each history determines one coverage vector. By
\Cref{cor:onehit} each coverage vector steers at most one candidate onto
target~2. Hence at most $2^{t-1}$ values of $j$ have $\tau=t$, so
$\Pr[\tau=t]\le 2^{t-1}/M$. Summing over $t\le u$ gives \Cref{eq:tau-tail},
and the tail-sum formula for $\min\{\tau-1,T\}$ gives \Cref{eq:count-final}.
The full proof is in \Cref{app:proof-count}.
\end{proof}

\subsection{The impossibility theorem}
\label{sec:impossibility}

Taking $M=2^T$ in \Cref{lem:count} makes the search too long for the
horizon, which is our main result for partial information.

\begin{theorem}[unknown $K$ with partial information is not learnable]
\label{thm:impossible}
Let $G_3$ be as in \Cref{eq:g3-def} and let
$\Univ=\{\alpha_{r_1,r_2}:0\le r_1<r_2\le 1\}$ be the infinite universe of
threshold types \Cref{eq:type-utils}. For every $T\ge 1$ and every
$\Kmax\ge 1$, every online algorithm for \Cref{prob:partial}, randomized
or deterministic, even one knowing $T$, $\Univ$, $\Kmax$ and the
construction, suffers, on some instance with $|K|=1$,
\begin{equation}
\Regret(T)\;\ge\;T-2+\frac{T+2}{2^{T}}\;>\;T-2 .
\label{eq:impossible}
\end{equation}
The game is fixed; only the adversary's choice of $K$ within it depends on $T$.
\end{theorem}

\begin{proof}[Proof sketch]
Set $M=2^{T}$, draw $j$ uniformly from $[M]$ and let $\alpha^{(j)}$ of
\Cref{eq:dyadic} attack in every round. The coverage vector whose first
coordinate is the midpoint of $I_j$ steers $\alpha^{(j)}$ in every round,
with both conditions of \Cref{eq:hit-iff} holding strictly, so the
benchmark earns $0$ and regret counts failed rounds. Under partial
information a deterministic algorithm chooses each coverage vector from the
targets attacked so far, so \Cref{lem:count} with $M=2^{T}$ applies directly
and yields \Cref{eq:impossible} on average over $j$. Averaging over the
algorithm's coins extends this to randomized algorithms, and some fixed $j$
attains the average. The full proof is in \Cref{app:proof-impossible}.
\end{proof}

\begin{remark}[the bound is exactly tight]
\label{rem:tight-lb}
The constant cannot be improved. The defender that keeps the set of indices
consistent with its observations, which by \Cref{eq:type-line} is an integer
interval, and plays $\pvec_t=(k/M,0,1-k/M)$ for its midpoint $k$ attains
regret exactly $T-2+(T+2)2^{-T}$ on average over $j$. Indeed, up to round $T$
every history still holds a nonempty interval, so each history catches
exactly one candidate and \Cref{eq:tau-tail} holds with equality;
\Cref{sec:exp-imp} confirms this in simulation. Since $\Regret(T)\le T$ always in $G_3$, the defender loses all
but two of the $T$ rounds.
\end{remark}

\citet{letchford2009} show that a defender facing a \emph{single}
unknown attacker type can learn a near-optimal commitment from a small number
of best-response queries, which at first sight is in tension with
\Cref{thm:impossible}. There is no contradiction. Any such query bound must
depend on the geometry of the best-response regions, and the family
$\Univ_M$ of \Cref{eq:dyadic} drives the region that has to be located down
to width $1/M=2^{-T}$. \citet{bollini2026} prove a related impossibility for
Bayesian Stackelberg games with three types drawn i.i.d.; there, payoffs of
polynomial precision suffice, whereas a single type, as in
\Cref{thm:impossible}, requires the universe to be exponentially large in
$T$. \Cref{lem:count} measures this cost of precision exactly. With the
finite universe $\Univ_M$ and a single type, summing
$\Pr[\tau>u]\ge 1-2^{u}/M$ from \Cref{eq:tau-tail} over
$u\le\lfloor\log_2 M\rfloor$ gives regret at least
$\min\{T,\lfloor\log_2 M\rfloor\}-2$, and the defender of
\Cref{rem:tight-lb} identifies the type within $\lceil\log_2 M\rceil$
rounds. A single type therefore costs $\Theta(\min\{T,\log M\})$, and
\Cref{thm:impossible} is the case of payoffs given to $T$ bits of
precision.

\subsection{Two types already cost \texorpdfstring{$\sqrt T$}{sqrt T}}
\label{sec:sqrtT-main}

\Cref{thm:impossible} is driven entirely by the universe being infinite.
With a single type from a finite universe $\Univ_M$, the defender of
\Cref{rem:tight-lb} identifies the type after $\lceil\log_2 M\rceil$ rounds
and loses nothing afterwards, so its regret does not grow with $T$. Two
types are different: even when the defender knows both of them and learns
the type after every round, it must keep guessing which one attacks next,
and this alone costs $\sqrt T$.

\begin{proposition}[two types already cost $\sqrt{T}$]
\label{prop:sqrtT}
Consider $G_3$ with the two types $\alpha^{(1)}=\alpha_{0,1/2}$ and
$\alpha^{(2)}=\alpha_{1/2,1}$, so that $|\Univ|=|K|=\Kmax=2$. For every
$T\ge 2$ and every online algorithm, randomized or deterministic, under
either feedback model, and even one told $K$ and the utilities of both types
before the first round, there is a sequence $\avec\in K^T$ with
\begin{equation}
\Regret(T)\;\ge\;\tfrac{1}{2}\sqrt{\lfloor T/2\rfloor}.
\label{eq:sqrtT}
\end{equation}
In fact the bound holds on average over $\avec$ drawn i.i.d.\ uniformly
from $K$, so some fixed sequence attains it.
\end{proposition}

\begin{proof}[Proof sketch]
Draw the attackers i.i.d.\ uniformly from $K$, and let $N_1$ and $N_2$
count the rounds of each type. The two intervals are disjoint, so by
\Cref{cor:onehit} a fixed coverage vector steers at most one type; the best
one steers the more frequent type and fails only on the rounds of the other,
so the benchmark earns $-\min\{N_1,N_2\}$. An algorithm, whatever it has
seen, cannot predict the next type, so its coverage vector steers the
attacker with probability at most $1/2$ in each round, and it earns at most
$-T/2$ in expectation. The gap is $\tfrac12\EE|N_1-N_2|$, half the expected
distance from the origin of a simple random walk after $T$ steps, which is
at least $\tfrac12\sqrt{\lfloor T/2\rfloor}$. The full proof is in
\Cref{app:proof-sqrtT}.
\end{proof}

So the $\sqrt T$ dependence of \Cref{thm:main} cannot be improved, even with
the type set known. Returning to \Cref{thm:impossible}: the defender failed
there because it could never learn the identity of its opponent. The next
section asks how little information about identities is enough to restore
sublinear regret.

\section{Block Revelation of Attacker Types: Sublinear Regret}
\label{sec:block}

\Cref{thm:impossible} shows that under partial information the defender
must receive some information about the identities of the attackers;
observing only the attacked target is not enough. This section studies a
minimal form of such information. Every $B$ rounds, an oracle gives the
defender the \emph{set} of attacker types seen so far, together with their
utility vectors. The oracle does not reveal when a type attacked, how often
it attacked, or whether it will attack again, so the defender still cannot
tell which type was responsible for any particular round. For example, an
anti-poaching unit may receive a monthly intelligence report listing the
poaching groups known to be active in a park, without saying which incident
each group was responsible for.

Identifying the attacker from its observed best responses is exactly the
difficulty considered by \citet[Section 6]{balcan2015}. Our additional
information does not remove this difficulty; it only tells the defender
which types are present. Our algorithm therefore runs a partial-information
learner for a known type set on the types reported so far, and its regret
is that learner's rate, summed over the stretches between new types, plus
an additive cost for the delay. With the learner of \citet{balcan2015} this
gives a bound of order $T^{2/3}$.

\begin{problem}[Unknown-$K$ repeated SSG with block revelation]
\label{prob:block}
Fix a \emph{revelation length} $B\in\NN$, known to the defender, and divide
the horizon into $N_B=\lceil T/B\rceil$ \emph{revelation blocks}
\begin{equation}
\cB_\nu
=
\{(\nu-1)B+1,\ \dots,\ \min\{\nu B,T\}\},
\qquad \nu\in[N_B].
\label{eq:blocks}
\end{equation}
The game is \Cref{prob:partial} with the following additional feedback: at
the end of $\cB_\nu$, the defender receives the set
\begin{equation}
D_\nu
=
\{a_t:\ t\le \min\{\nu B,T\}\}
\subseteq K
\label{eq:Dnu}
\end{equation}
of distinct types that have attacked in the first $\nu$ blocks, together
with the utility vectors of all types in $D_\nu$. We set $D_0=\emptyset$.
No other information is revealed, and within each round the feedback
remains partial.
\end{problem}

This model is weaker than \Cref{prob:unknownK}: even at $B=1$ the defender
learns a round's type only when that type is new. Taking $B\ge T$ makes the
revelation useless and recovers \Cref{prob:partial}, so by
\Cref{thm:impossible} some bound on $B$ is necessary.

\subsection{The algorithm}
\label{sec:block-alg}

\Cref{alg:bgc} keeps the outer skeleton of \Cref{alg:gc}: it maintains a
catalog $C$ of the types the defender has been told about, runs an inner
learner on the experts $\cE(C;\epsilon)$, and restarts the inner learner
whenever $C$ grows. It differs from \Cref{alg:gc} in two places. First, the
catalog can grow only at the checkpoints $t=B,2B,\dots$, so restarts happen
at block boundaries rather than at the exact round a new type first
attacks. Second, the full-information learner $\mathcal{A}(C)$ cannot be
used: it updates every expert with the loss $-U_d(b_{a_t}(\pvec),\pvec)$,
which requires knowing $a_t$. We therefore use the partial-information
learner of \citet[Algorithm 1]{balcan2015}, which we denote $\Api(C)$ and
which needs only the attacked target.

We treat $\Api(C)$ as a black box. Its input is a catalog $C$, which fixes
the expert set $\cE(C;\epsilon)$; at each round it receives the observed
target $i_t$ and outputs the coverage vector $\pvec_t$ to play. Internally
it groups its rounds into estimation blocks, which are unrelated to the
revelation blocks $\cB_\nu$, explores on a few rounds of each block, uses
them to form an unbiased estimate of every expert's average payoff on the
block, and feeds that estimate to an expert algorithm whose output drives
the next block. \Cref{app:balcan-alg} describes it in more detail. The
analysis uses its regret bound, which \citet{balcan2015} prove for a known
set $K$ of types, together with one further fact shown in the proof of
\Cref{lem:anytime}: the bound also holds at every prefix of a run.

\begin{theorem}[{\citet[Theorem 6.1]{balcan2015}}]
\label{thm:balcan61}
Given a repeated security game with partial-information feedback, $n$
targets, a known set $K$ of attacker types, and time horizon $T$, there is
an online algorithm that for any oblivious sequence of attackers $\avec$
plays at time $t$ a randomly chosen coverage vector $\pvec_t$ and satisfies
\[
\EE\Bigl[\sum_{t=1}^{T}U_d(b_{a_t}(\pvec_t),\pvec_t)\Bigr]
\;\ge\;\sum_{t=1}^{T}U_d(b_{a_t}(\pvec^{*}),\pvec^{*})
-O\bigl(T^{2/3}\,n|K|\log^{1/3}(n|K|)\bigr),
\]
where the expectation is over the algorithm's internal randomization.
\end{theorem}

In our use the known set is the catalog $C$ and the algorithm is $\Api(C)$.
Its horizon must be known in advance; \Cref{lem:anytime} below removes this
requirement.

\begin{algorithm}[t]
\caption{\textsc{Block-Grow-Catalog} (partial information, revelation length $B$).}
\label{alg:bgc}
\begin{algorithmic}[1]
\small
\Require horizon $T$; revelation length $B$; game primitives; default strategy $\pvec_{\text{def}}\in\cP$
\State $C\gets\emptyset$;~$\epsilon\gets 1/T$;~$\Api\gets\textsc{null}$
   \Comment{empty catalog}
\For{$t=1,\dots,T$}
  \State \textbf{Play} $\pvec_t\gets\pvec_{\text{def}}$ if $C=\emptyset$, else the strategy $\Api$ prescribes for round $t$
  \State \textbf{Observe} attacked target $i_t=b_{a_t}(\pvec_t)$
  \State pass $i_t$ to $\Api$ as its round-$t$ feedback
  \If{$t=\nu B$ for some $\nu$, or $t=T$}
    \State receive $D_\nu\subseteq K$ and the utilities of each $\alpha\in D_\nu$
    \If{$D_\nu\ne C$} \Comment{new type(s) appeared during $\cB_\nu$}
      \State $C\gets D_\nu$
      \State $\Api\gets\Api(C)$
         \Comment{restart on enlarged catalog}
    \EndIf
  \EndIf
\EndFor
\end{algorithmic}
\end{algorithm}

\subsection{Discovery blocks and epochs}
\label{sec:block-structure}

\begin{figure}[t]
\centering
\begin{tikzpicture}[x=0.55cm,y=0.42cm]
\foreach \i/\lbl/\sh in {0/{$\cB_1$}/black!28,%
                        1/{$\cB_2$}/black!5,%
                        2/{$\cB_3$}/black!5,%
                        3/{$\cB_4$}/black!28,%
                        4/{$\cB_5$}/black!5,%
                        5/{$\cB_6$}/black!5,%
                        6/{$\cB_7$}/black!5,%
                        7/{$\cB_8$}/black!28} {
  \draw[fill=\sh] (\i,0) rectangle (\i+1,1);
  \node[font=\scriptsize] at (\i+0.5,0.5) {\lbl};
}
\node[font=\scriptsize,anchor=south] at (0.5,1.05) {disc.};
\node[font=\scriptsize,anchor=south] at (3.5,1.05) {disc.};
\node[font=\scriptsize,anchor=south] at (7.5,1.05) {disc.};
\draw[thick] (1,-0.15) -- (1,-0.55) -- (3,-0.55) -- (3,-0.15);
\node[font=\scriptsize,anchor=north] at (2,-0.60) {epoch $\Ep_1$};
\draw[thick] (4,-0.15) -- (4,-0.55) -- (7,-0.55) -- (7,-0.15);
\node[font=\scriptsize,anchor=north] at (5.5,-0.60) {epoch $\Ep_2$};
\draw[->,thin] (-0.3,-0.05) -- (8.4,-0.05);
\node[font=\scriptsize,anchor=west] at (8.4,-0.05) {\;$t$};
\end{tikzpicture}
\caption{The block decomposition. Shaded blocks are the \emph{discovery
blocks}: their closing checkpoint enlarges the catalog. The run of
unshaded blocks strictly between two consecutive discovery blocks forms an
\emph{epoch}; inside an epoch the catalog is fixed and \Cref{alg:bgc} runs
one uninterrupted inner learner.}
\label{fig:blocks}
\end{figure}
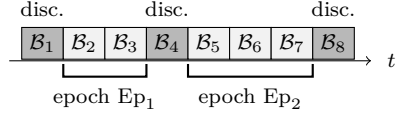

The analysis cuts the horizon into two kinds of stretches, shown in
\Cref{fig:blocks}: \emph{discovery blocks}, during which a type the defender
has not yet been told about may attack, and \emph{epochs}, the runs of
blocks between them, during which every attacker is already in the catalog.
The proof charges discovery blocks the maximum regret and applies the
learner's guarantee on each epoch.

\begin{definition}[discovery blocks and epochs]
\label{def:blocks}
A revelation block $\cB_\nu$ is a \emph{discovery block} if its report
names a type not named before, that is, $D_\nu\ne D_{\nu-1}$. Let
$\cB_{\beta_1},\dots,\cB_{\beta_m}$ be the discovery blocks, in order. The
$j$-th \emph{epoch} $\Ep_j$ consists of the blocks strictly between
$\cB_{\beta_j}$ and $\cB_{\beta_{j+1}}$, or up to the end of the horizon
when $j=m$. We write $L_j$ for the number of rounds in $\Ep_j$ and
$C_j\coloneqq D_{\beta_j}$ for the catalog during it.
\end{definition}

\begin{lemma}[decomposition]
\label{lem:blocks}
Let $K'=\{a_1,\dots,a_T\}$ be the set of types appearing in $\avec$.
\begin{enumerate}[leftmargin=1.4em,itemsep=0.1em,topsep=0.2em,label=(\alph*)]
  \item There are $m\le |K'|$ discovery blocks, and together with the
  epochs they partition $[1,T]$; in particular $\sum_j L_j\le T$.
  \item Throughout $\Ep_j$ the catalog is $C_j$, every attacker belongs
  to $C_j$, and a single instance of $\Api(C_j)$ runs, started at the first
  round of $\Ep_j$.
  \item The discovery blocks and the epochs depend only on $\avec$, not on
  the defender's play.
\end{enumerate}
\end{lemma}

\begin{proof}
(a)~The first block is a discovery block, since $D_0=\emptyset$, and every
discovery block names at least one of the $|K'|$ types for the first time.
(b)~A block that is not a discovery block names no new type, so every type
attacking in it had already been named; and \Cref{alg:bgc} changes its
catalog, and restarts its learner, only after a discovery block.
(c)~Whether a block is a discovery block depends only on which types
attack, and the oblivious adversary fixes these before play begins.
\end{proof}

Part~(c) is where the obliviousness of the adversary is used: it makes the
attackers within an epoch an oblivious sequence, as the learner's guarantee
requires. What the learner does during discovery blocks does not matter,
since the proof charges those rounds the maximum regret.

The proof of \Cref{thm:block} applies \Cref{thm:balcan61} on each epoch.
That theorem needs its horizon in advance, whereas an epoch ends only when
a new type appears, so its length is unknown when it starts. The following
lemma removes this requirement.

\begin{lemma}[anytime partial-information learner]
\label{lem:anytime}
For every finite catalog $C$ there is a variant of $\Api(C)$ that does not
require its horizon in advance, and an absolute constant $c_4$, such that on
any oblivious sequence $a'_1,a'_2,\dots$ with $a'_t\in C$ and for every
$L\le T$ simultaneously,
\begin{equation}
\max_{\pvec\in\cP}\sum_{t=1}^{L}U_d(b_{a'_t}(\pvec),\pvec)
 -\EE\Bigl[\sum_{t=1}^{L}U_d(b_{a'_t}(\pvec_t),\pvec_t)\Bigr]
\;\le\;c_4\,L^{2/3}\,n\,|C|\,\log^{1/3}(n|C|).
\label{eq:anytime}
\end{equation}
\end{lemma}

The benchmark on the left is the best fixed coverage vector for the
sequence $a'_1,\dots,a'_L$ alone; the proof of \Cref{thm:block} passes from
it to the global benchmark $\pvec^{*}$. The variant is built with the
doubling trick of \citet{auer1995} \citep[see also][Sec.~2.3]{cesa2006},
which turns a learner tuned to a known horizon into one that needs none.

\begin{proof}[Proof sketch]
\emph{(i) A run tuned to a horizon $H$ obeys its bound at every prefix.}
The proof of \Cref{thm:balcan61} bounds the regret of $\Api(C)$ by two error
terms: the payoff lost on exploration rounds, and the regret of its internal
expert algorithm on the estimated payoffs. Stopping that argument after
$L\le H$ rounds bounds the regret at round $L$ by the same two terms
evaluated at $L$ (a block cut in the middle adds only a lower-order term).
Both terms are nondecreasing in $L$, so at every prefix they are at most
their values at $H$: the regret at any round $L\le H$ is at most
$c_0H^{2/3}\phi$, the bound of \Cref{thm:balcan61}, where
$\phi\coloneqq n|C|\log^{1/3}(n|C|)$.
\emph{(ii) Doubling.} The variant runs fresh copies of $\Api(C)$ tuned to
horizons $1,2,4,\dots$, one after another, so it needs no horizon. Given
$L$, let $S$ be the stage containing round $L$; the earlier stages fill
$2^{S}-1$ rounds, so $2^{S}\le L$. Only stage $S$ can be cut off before its
end, and by (i) its bound still holds. The best fixed vector over $[1,L]$
earns at most the sum of the stages' best vectors, so the regret over
$[1,L]$ is at most the sum of the stage bounds,
\[
c_0\phi\sum_{s=0}^{S}2^{2s/3}\;\le\;\frac{2^{2/3}}{2^{2/3}-1}\,c_0\phi\,2^{2S/3}
\;\le\;2.71\,c_0\phi\,L^{2/3},
\]
which is \Cref{eq:anytime} with $c_4=2.71\,c_0$. The full proof is in
\Cref{app:anytime}.
\end{proof}

With these pieces in place we can state our main result for block
revelation, in which \Cref{alg:bgc} is run with the learner of
\Cref{lem:anytime}.

\begin{theorem}[block revelation restores no-regret]
\label{thm:block}
Let the defender play \Cref{alg:bgc} in \Cref{prob:block}, and let
$K'=\{a_1,\dots,a_T\}$ be the set of types appearing in the oblivious
sequence $\avec\in K^T$. Then
\begin{equation}
\Regret(T)\;\le\;
c\, n\,|K'|^{4/3}\,T^{2/3}\log^{1/3}(n|K'|)\;+\;2|K'|B
\label{eq:block-main}
\end{equation}
for an absolute constant $c$; in particular
$\Regret(T)=O(n\Kmax^{4/3}T^{2/3}\log^{1/3}(n\Kmax)+\Kmax B)$.
\end{theorem}

\begin{proof}
Adopt the notation of \Cref{def:blocks,lem:blocks}. For a set of rounds
$S\subseteq[1,T]$ let
\begin{equation}
\Regret(S)=\sum_{t\in S}U_d(b_{a_t}(\pvec^{*}),\pvec^{*})
-\EE\Bigl[\sum_{t\in S}U_d(b_{a_t}(\pvec_t),\pvec_t)\Bigr]
\label{eq:regret-of-S}
\end{equation}
denote the regret incurred on the rounds in $S$; it is additive over disjoint
sets of rounds. The idea is to split the rounds into two kinds, bound the
regret on each kind separately, and add.

\emph{Splitting the rounds.} By \Cref{lem:blocks}(a), the horizon is the
disjoint union of the $m\le |K'|$ discovery blocks $\cB_{\beta_j}$ and the
$m$ epochs $\Ep_j$, so by additivity
\begin{equation}
\Regret(T)=\sum_{j=1}^{m}\Regret(\cB_{\beta_j})+\sum_{j=1}^{m}\Regret(\Ep_j).
\label{eq:block-split}
\end{equation}

\emph{Discovery blocks.} Inside a discovery block the defender may face a type
it has not yet been told about, so we make no use of what the learner does
there and charge each round the largest possible regret. Since all payoffs
lie in $[-1,1]$, a round contributes at most $2$ to \Cref{eq:regret-of-S};
each discovery block holds at most $B$ rounds, and there are at most $|K'|$
of them, so
\begin{equation}
\sum_{j=1}^{m}\Regret(\cB_{\beta_j})\;\le\;2|K'|B.
\label{eq:disc-block-sum}
\end{equation}

\emph{Epochs.} Fix an epoch $\Ep_j$. By \Cref{lem:blocks}(b), a single
instance of $\Api(C_j)$ runs throughout it and every attacker in it belongs
to $C_j$; by \Cref{lem:blocks}(c), the attackers in it form an oblivious
sequence. These are exactly the hypotheses of \Cref{lem:anytime}, which
moreover does not require the epoch length $L_j$ to be known in advance. The
lemma bounds the regret against the best fixed coverage vector \emph{for
this epoch}; since $\pvec^{*}$ is one of the candidates in that maximum, the
same bound holds against $\pvec^{*}$. Using also $|C_j|\le |K'|$, as the
catalog only ever contains types that appear,
\begin{equation}
\Regret(\Ep_j)\;\le\;c_4\,n\,|C_j|\log^{1/3}(n|C_j|)\,L_j^{2/3}
\;\le\;c_4\,n\,|K'|\log^{1/3}(n|K'|)\,L_j^{2/3}.
\label{eq:combined-epoch-block}
\end{equation}

\emph{Adding up.} The adversary decides when new types arrive and so controls
the epoch lengths, subject only to $m\le |K'|$ and $\sum_j L_j\le T$. By
H\"older's inequality with exponents $3$ and $3/2$, applied to the vectors
$(1,\dots,1)$ and $(L_1^{2/3},\dots,L_m^{2/3})$,
\begin{equation}
\sum_{j=1}^{m}L_j^{2/3}
\;\le\;m^{1/3}\Bigl(\sum_{j=1}^{m}L_j\Bigr)^{2/3}
\;\le\;|K'|^{1/3}\,T^{2/3},
\label{eq:holder}
\end{equation}
with equality when all epochs have the same length: the adversary does the
most damage by spreading the arrivals of new types evenly. Substituting
\Cref{eq:disc-block-sum,eq:combined-epoch-block,eq:holder} into
\Cref{eq:block-split} gives
\[
\Regret(T)\;\le\;c_4\,n\,|K'|^{4/3}\,T^{2/3}\log^{1/3}(n|K'|)+2|K'|B,
\]
which is \Cref{eq:block-main} with $c=c_4$; using $|K'|\le\Kmax$ gives its
$\Kmax$ form. The exponent $4/3$ collects a factor $|K'|$ from the size of
the catalog inside each epoch and a factor $|K'|^{1/3}$ from summing over at
most $|K'|$ epochs.
\end{proof}

The proof uses only one property of the inner learner: an anytime bound on
its regret within each epoch. A learner with a better rate therefore gives a
better bound. If the per-epoch regret is at most $g(|C|)\sqrt{L}$, the same
argument, with the Cauchy--Schwarz inequality in place of H\"older's,
gives $\Regret(T)\le g(|K'|)\sqrt{|K'|T}+2|K'|B$. The learner of
\citet{bernasconi2023} attains $\tilde O(\mathrm{poly}(|K|)\sqrt T)$ for a
known type set under partial information, so if it admits an anytime
version, block revelation costs only the additive delay term on top of a
$\sqrt T$ rate. We state \Cref{thm:block} with the learner of
\citet{balcan2015} because its anytime version is the one we prove
(\Cref{lem:anytime}).

\subsection{How rarely reports are needed}
\label{sec:reading}

The bound \Cref{eq:block-main} has three ingredients. The exponent $2/3$ is
that of the partial-information learner of \Cref{thm:balcan61}, not of the
model; the
extra factor $|K'|^{1/3}$ is the price of learning the catalog gradually,
since the learner is restarted once for each new type; and the additive
$2|K'|B$ is the price of hearing about new types only at the end of a block.
Because \Cref{alg:bgc} never reads $\Kmax$, the bound depends only on the
types that actually appear.

Block revelation sits between the two feedback models of the earlier
sections. As $B$ grows it approaches the partial-information model of
\Cref{sec:partial}, where no algorithm learns, and at $B=1$ it reduces to
partial information with a known type set, apart from the rounds at which
new types first appear.

Reports are costly in practice, since each is a piece of intelligence that
someone has to gather, so it matters how rarely they can arrive. The delay
is free as long as the batching term $2|K'|B$ stays below the learning term.

\begin{corollary}[how rarely reports are needed]
\label{cor:B}
If
\begin{equation}
B\;\le\;B^{*}\coloneqq\tfrac{1}{2}\,c\,n\,\Kmax^{1/3}\,T^{2/3}\log^{1/3}(n\Kmax),
\label{eq:Bstar}
\end{equation}
then $\Regret(T)=O\bigl(n\Kmax^{4/3}T^{2/3}\log^{1/3}(n\Kmax)\bigr)$, the
same rate as with a report after every round. The defender therefore needs
only
\begin{equation}
\lceil T/B^{*}\rceil\;=\;O\!\left(\frac{T^{1/3}}{n\,\Kmax^{1/3}\log^{1/3}(n\Kmax)}\right)
\label{eq:ncalls}
\end{equation}
reports over the whole horizon.
\end{corollary}

The threshold $B^{*}$ guides whoever schedules the reports; the defender
itself need not know it, and \Cref{alg:bgc} never uses it.

\begin{proof}
For $B\le B^{*}$ the batching term satisfies $2|K'|B\le2\Kmax B^{*}$, which
is the learning term of \Cref{eq:block-main} evaluated at $|K'|=\Kmax$, so the
bound at most doubles. The number of reports is $\lceil T/B\rceil$, which is
smallest at $B=B^{*}$.
\end{proof}

For example, with $T=10^{6}$ rounds and a handful of types,
$T^{2/3}=10^{4}$, so, ignoring the factors in $n$ and $\Kmax$, a report
every $10{,}000$ rounds, about $100$ reports in all, gives the same rate as
a report after every round.

\subsection{The additive term is necessary}
\label{sec:blocklb-main}

The additive term of \Cref{thm:block} is not an artifact of the analysis:
no algorithm can avoid a cost of order $\Kmax B$.

\begin{theorem}[a matching lower bound in $B$]
\label{thm:blocklb}
Let $G_3$ and $\Univ$ be as in \Cref{thm:impossible}. For every $B\ge 1$,
every $\Kmax\ge 1$, every $T\ge \Kmax B$ and every online algorithm for
\Cref{prob:block}, there is an instance with $|K|=\Kmax$ on which
\begin{equation}
\Regret(T)\;\ge\;\Kmax\,(B-2).
\label{eq:blocklb}
\end{equation}
\end{theorem}

\begin{proof}[Proof sketch]
The adversary uses the game $G_3$ of \Cref{sec:gadget} and picks $\Kmax$
types that one fixed coverage vector would steer onto target~2 all at once,
so the best fixed strategy loses nothing and every failed round counts as
regret. The types arrive one per revelation block, and each new type is
drawn at random from $2^B$ candidates; the defender learns which one it was
only when the oracle names it at the end of that block. Inside the block
the defender must therefore find the new type by trial and error from the
attacked targets alone, and by \Cref{lem:count} such a search over $2^B$
candidates fails in more than $B-2$ of its $B$ rounds in expectation over
the draw. This cost is paid once for each of the $\Kmax$ types. Since the
bound holds on average over the random draw, some fixed draw attains it,
and that fixed draw is the instance in the statement. The full proof is in
\Cref{app:blocklb}.
\end{proof}

Together with \Cref{thm:block}, the cost of batching is $\Theta(\Kmax B)$
for $B\ge 3$. \Cref{app:blocklb} compares the two bounds and shows that, in
the regime where the batching term dominates, they differ by a factor of at
most $2B/(B-2)$, up to lower-order terms. At $\Kmax=1$ and $B=T$ the lower
bound reads $T-2$, so \Cref{thm:impossible} is the limiting case of
\Cref{thm:blocklb} in which no report arrives in time. We now step back and
summarize what the three settings show together.

\section{Conclusions and Future Work}
\label{sec:conc}

We studied repeated Stackelberg security games in which the defender does
not know the set of attacker types, only that it is finite, and showed that
learnability then hinges on what the defender observes. Seeing the type of
each attacker, \textsc{Grow-Catalog} attains $\tilde O(n\Kmax\sqrt T)$
regret without knowing $\Kmax$, the known-type rate up to a factor
$\sqrt{\Kmax}$. Seeing only the attacked target, no algorithm achieves
sublinear regret, even against a single type, because the defender cannot
identify its opponent. A periodic list of the types seen so far is enough to
restore sublinear regret, $\tilde O(n\Kmax^{4/3}T^{2/3}+\Kmax B)$, with an
additive cost of the delay that is unavoidable, and only $O(T^{1/3})$ such
reports are needed.

Several questions remain open. Neither the factor $\sqrt{\Kmax}$ of
\Cref{thm:main} nor the factor $\Kmax^{1/3}$ of \Cref{thm:block} is known to
be necessary. It is also open whether the $\sqrt T$ learner of
\citet{bernasconi2023} admits an anytime version, which would bring block
revelation to a $\sqrt T$ rate plus the delay term, and, since a single type
costs only $\Theta(\log M)$ at precision $1/M$, whether a few types already
force polynomial regret in security games at polynomial payoff precision, as
three do in the general games of \citet{bollini2026}. Finally, our adversary is oblivious;
extending the results to adaptive attackers, or to reports that are noisy or
incomplete, would bring the model closer to the intelligence that defenders
actually receive.

\bibliography{refs}

\clearpage
\appendix

\section{Proofs for \texorpdfstring{\Cref{sec:partial}}{Section 4}}
\label{app:partial-proofs}

\subsection{Proof of \texorpdfstring{\Cref{lem:hit}}{Lemma} and \texorpdfstring{\Cref{cor:onehit}}{Corollary}}
\label{app:proof-hit}

\paragraph{\Cref{lem:hit}.} Ties are broken toward the lowest-indexed target.
For target~2 to be selected it must therefore be \emph{strictly} better than
target~1, since a tie is resolved in favor of target~1, but only
\emph{weakly} better than target~3, since a tie is resolved in favor of
target~2. By \Cref{eq:type-U}, these conditions are respectively
\[
U_\alpha(2,\pvec)>U_\alpha(1,\pvec)
\iff 0>r_1-p_1
\iff p_1>r_1
\]
and
\[
U_\alpha(2,\pvec)\ge U_\alpha(3,\pvec)
\iff 0\ge(1-r_2)-p_3
\iff p_3\ge1-r_2.
\]
Conversely, if both conditions hold, then target~2 beats target~1 strictly
and ties with or beats target~3, so lowest-index tie-breaking selects
target~2. This proves \Cref{eq:hit-iff}. For the second claim, the
single-resource constraint $p_1+p_2+p_3\le1$ together with $p_2\ge 0$ gives
\[
p_1\le1-p_3\le1-(1-r_2)=r_2,
\]
where the second inequality uses $p_3\ge1-r_2$. Together with $p_1>r_1$, this
gives $p_1\in(r_1,r_2]$.\hfill$\square$

\paragraph{\Cref{cor:onehit}.} Suppose two of the types, say $\alpha^{(j)}$
and $\alpha^{(j')}$, both strike target~2 under the same coverage vector
$\pvec$. The second assertion of \Cref{lem:hit}, applied to each of them,
places the single number $p_1$ in the interval of $\alpha^{(j)}$ and in the
interval of $\alpha^{(j')}$. Since the intervals are pairwise disjoint, this
forces $j=j'$.\hfill$\square$

\subsection{Proof of \texorpdfstring{\Cref{lem:count}}{Lemma}}
\label{app:proof-count}

The proof rests on a counting argument: we bound, for each round $t$
separately, how many of the $M$ candidate indices can be steered onto
target~2 for the very first time at that round. Since the events
$\{\tau=t\}$ for different $t$ are disjoint, summing those bounds gives the
tail bound \Cref{eq:tau-tail}, and a standard identity for the expectation of
a nonnegative integer random variable then gives \Cref{eq:count-final}.
Write $\tau(j)$ for the value of $\tau$ when the hidden index is $j$.

\emph{Step 1: how many histories are possible.} Fix a round $t$ and
consider the indices $j$ with $\tau(j)=t$. For such an index the attacker
struck target~1 or target~3 in each of the rounds $1,\dots,t-1$; had it
struck target~2, we would have $\tau(j)<t$. The defender's record of the
past is the string $(i_1,\dots,i_{t-1})$ of attacked targets, so at the start
of round $t$ it is in one of at most $2^{t-1}$ histories. Write $H_t$ for
this set of histories, so $|H_t|\le 2^{t-1}$.

\emph{Step 2: each history determines one coverage vector.} By hypothesis
the coverage vector played in round $t$ is a fixed function $\pvec_t(h)$ of
the history $h$. In particular, two indices that produce the same history
lead the defender to play the same coverage vector.

\emph{Step 3: each coverage vector steers at most one candidate.} The
intervals are pairwise disjoint, so by \Cref{cor:onehit} at most one index
$j'$ has $b_{\alpha^{(j')}}(\pvec_t(h))=2$.

\emph{Step 4: combining the three.} Let $j$ be any index with $\tau(j)=t$.
By Step~1 it leads the defender into some history $h\in H_t$; by Step~2 the
defender then plays $\pvec_t(h)$; and by the definition of $\tau$ the
attacker $\alpha^{(j)}$ strikes target~2 under it. By Step~3, at most one
index does so. Hence distinct indices with $\tau=t$ arise from distinct
histories, and
\[
|\{j:\tau(j)=t\}|\;\le\;|H_t|\;\le\;2^{t-1}.
\]
Since $j$ is uniform on $[M]$, dividing by $M$ turns this count into a
probability, $\Pr[\tau=t]\le 2^{t-1}/M$. The first few values are
$\Pr[\tau=1]\le 1/M$, $\Pr[\tau=2]\le 2/M$ and $\Pr[\tau=3]\le 4/M$.

\emph{Step 5: the tail bound.} The events $\{\tau=t\}$, $t=1,\dots,u$, are
pairwise disjoint, so their probabilities add, and summing the geometric
series $1+2+\cdots+2^{u-1}=2^u-1$ gives
\[
\Pr[\tau\le u]\;=\;\sum_{t=1}^{u}\Pr[\tau=t]\;\le\;\frac{1}{M}\sum_{t=1}^{u}2^{t-1}
\;=\;\frac{2^u-1}{M},
\]
which is \Cref{eq:tau-tail}.

\emph{Step 6: from the tail bound to the expectation.} Put
$X=\min\{\tau-1,T\}$, a random variable with values in $\{0,1,\dots,T\}$. For
such a variable the tail-sum formula gives $\EE[X]=\sum_{s=0}^{T-1}\Pr[X>s]$.
Now $X>s$ holds exactly when $\tau>s+1$, so writing $u=s+1$,
\[
\EE[X]\;=\;\sum_{u=1}^{T}\Pr[\tau>u]\;=\;\sum_{u=1}^{T}\bigl(1-\Pr[\tau\le u]\bigr)
\;\ge\;T-\sum_{u=1}^{T}\frac{2^{u}-1}{M},
\]
where the inequality substitutes \Cref{eq:tau-tail}. Finally
$\sum_{u=1}^{T}2^u=2^{T+1}-2$ and $\sum_{u=1}^{T}1=T$, so the subtracted sum
equals $(2^{T+1}-2-T)/M$, which is \Cref{eq:count-final}.\hfill$\square$

\subsection{Proof of \texorpdfstring{\Cref{thm:impossible}}{Theorem}}
\label{app:proof-impossible}

Set $M=2^{T}$ and take the family $\Univ_M\subset\Univ$ of \Cref{eq:dyadic},
whose $M$ intervals each have width $2^{-T}$. For $j\in[M]$, the instance
$j$ is $K=\{\alpha^{(j)}\}$ with $a_t=\alpha^{(j)}$ in every round. It is
admissible for every $\Kmax\ge1$, since $|K|=1$, and the adversary is
oblivious. Fix an online algorithm, write $\Regret_j(T)$ for its regret on
instance $j$, and let $\overline{R}=\frac1M\sum_{j=1}^M\Regret_j(T)$ be the
average over the instances. We show that $\overline{R}$ is at least the
right-hand side of \Cref{eq:impossible}; then some instance $j$ has
$\Regret_j(T)\ge\overline{R}$, which is the claim.

\emph{Part 1: the benchmark earns zero, so regret counts failed rounds.}
On instance $j$ let $\pvec^{*}=\bigl(\tfrac{2j-1}{2M},0,1-\tfrac{2j-1}{2M}\bigr)$,
whose first coordinate is the midpoint of $I_j$. Writing $r_1=(j-1)/M$ and
$r_2=j/M$ for the endpoints of $I_j$, we have $p^{*}_1>r_1$ and
$p^{*}_3=1-p^{*}_1>1-r_2$, both strictly. By \Cref{eq:hit-iff},
$b_{\alpha^{(j)}}(\pvec^{*})=2$ whatever the tie-breaking rule, so
$\pvec^{*}$ earns $0$ in every round. As $U_d\le 0$ everywhere by
\Cref{eq:g3-payoff}, no coverage vector does better, so $\pvec^{*}$ is a
maximizer in \Cref{eq:pstar} and the benchmark term of \Cref{eq:regret}
vanishes. Every other round is worth $-1$, so
\begin{equation}
\Regret_j(T)=\EE\bigl[\#\{t\le T:\text{the defender fails at }t\}\bigr],
\label{eq:regret-is-misses}
\end{equation}
the expectation being over the algorithm's coins.

\emph{Part 2: deterministic algorithms.} Suppose first that the algorithm is
deterministic. Under partial information its only feedback is the attacked
target, so its coverage vector in round $t$ is a function of
$i_1,\dots,i_{t-1}$, which is the setting of \Cref{lem:count}. Every round
before the first success $\tau$ is a failure, so the count in
\Cref{eq:regret-is-misses} is at least $\min\{\tau-1,T\}$. Averaging over
$j$ uniform on $[M]$ and applying \Cref{eq:count-final} with $M=2^T$,
\[
\overline{R}\ \ge\ \EE_j\bigl[\min\{\tau-1,T\}\bigr]
\ \ge\ T-\frac{2^{T+1}-T-2}{2^{T}}
\ =\ T-2+\frac{T+2}{2^{T}}.
\]

\emph{Part 3: randomized algorithms.} A randomized algorithm is a
distribution over deterministic ones, and its coins are independent of $j$.
The benchmark of instance $j$ does not depend on the coins, so
$\Regret_j(T)$ is the average over the coins of the regret of the
deterministic algorithm they select, and so is $\overline{R}$. Each
deterministic algorithm satisfies the bound of Part~2, hence so does their
average. This is the easy direction of Yao's principle \citep{yao1977}, and
it completes the proof.\hfill$\square$

\subsection{Proof of \texorpdfstring{\Cref{prop:sqrtT}}{Proposition}}
\label{app:proof-sqrtT}

Draw $a_1,\dots,a_T$ i.i.d.\ uniformly from $K$, fix an online algorithm,
and let $\overline{R}$ be its regret averaged over this draw. It suffices to
show $\overline{R}\ge\tfrac12\sqrt{\lfloor T/2\rfloor}$, since then some
fixed sequence has at least this regret. The two types have intervals
$(0,\tfrac12]$ and $(\tfrac12,1]$, which are disjoint, so by
\Cref{cor:onehit} any single coverage vector steers at most one of them onto
target~2. Write $N_i=|\{t:a_t=\alpha^{(i)}\}|$ for the number of rounds in
which type $i$ attacks, so that $N_1+N_2=T$.

\emph{Benchmark.} A fixed coverage vector steers at most one of the two
types, so on the rounds of the other type it fails and earns $-1$. Its total
payoff is therefore at most $-\min\{N_1,N_2\}$, and this is achieved: the
vector $(\tfrac14,0,\tfrac34)$ has $p_1=\tfrac14\in(0,\tfrac12]$ and
$p_3=\tfrac34\ge 1-\tfrac12$, so by \Cref{eq:hit-iff} it steers
$\alpha^{(1)}$ on all $N_1$ of its rounds, while $(\tfrac34,0,\tfrac14)$
steers $\alpha^{(2)}$ on all $N_2$ of its rounds. The better of the two
attains $-\min\{N_1,N_2\}$, so the benchmark of \Cref{eq:pstar} has value
exactly $-\min\{N_1,N_2\}$.

\emph{Algorithm.} Fix a round $t$ and condition on the history up to that
round. The history is a function of $a_1,\dots,a_{t-1}$ and of the
algorithm's own coins, and the sequence is drawn i.i.d., so the history
carries no information about $a_t$; in particular the coverage vector
$\pvec_t$, which is determined by the history, is independent of $a_t$. By
\Cref{cor:onehit}, $\pvec_t$ steers at most one of the two types, and $a_t$
equals that type with probability exactly $\tfrac12$. Hence the probability
of success in round $t$ is at most $\tfrac12$ conditionally on any history,
and therefore unconditionally. By linearity of expectation over the $T$
rounds, the algorithm's expected payoff is at most $-T/2$.

\emph{Combining.} Subtracting the algorithm's expected payoff from the
benchmark's gives
\[
\overline{R}\;\ge\;\tfrac{T}{2}-\EE[\min\{N_1,N_2\}].
\]
The identity $\min\{a,b\}=\tfrac12\bigl((a+b)-|a-b|\bigr)$ with $a=N_1$,
$b=N_2$ and $N_1+N_2=T$ gives $\min\{N_1,N_2\}=\tfrac12\bigl(T-|N_1-N_2|\bigr)$,
so
\[
\overline{R}\;\ge\;\tfrac{1}{2}\EE\bigl[|N_1-N_2|\bigr].
\]
Each $a_t$ is an independent fair coin flip between the two types, so
$N_1-N_2$ is the position after $T$ steps of a simple symmetric random walk
on the integers, and \Cref{lem:walk} below gives
$\EE|N_1-N_2|\ge\sqrt{\lfloor T/2\rfloor}$, which proves \Cref{eq:sqrtT}.
\hfill$\square$

The bound is tight in a strong sense: the success probability $\tfrac12$ is
attained by any algorithm that always plays inside one of the two intervals,
so all of them, follow the leader and an algorithm simply told $K$ alike,
have regret exactly $\tfrac12\EE|N_1-N_2|=\sqrt{T/(2\pi)}\,(1+o(1))$ on
average over the draw, the value \Cref{fig:impossible}(b) measures.

\begin{lemma}[random walk]
\label{lem:walk}
Let $S_T=\sum_{t\le T}\varepsilon_t$ with $\varepsilon_t\in\{-1,+1\}$
independent and uniform. Then $\EE|S_T|\ge\sqrt{\lfloor T/2\rfloor}$.
\end{lemma}
\label{app:walk}

\begin{proof}
\emph{Reduction to even length.} Since $|\cdot|$ is convex and $S_m$ is a
martingale, $|S_m|$ is a submartingale, so $\EE|S_m|$ is nondecreasing in
$m$. With $\nu=\lfloor T/2\rfloor$ we have $2\nu\le T$, hence
$\EE|S_T|\ge\EE|S_{2\nu}|$.

\emph{Closed form.} For a walk of $2\nu$ steps the classical identity
$\EE|S_{2\nu}|=2\nu\binom{2\nu}{\nu}4^{-\nu}$ holds \citep{feller1968}, so it
suffices to show
\begin{equation}
\binom{2\nu}{\nu}4^{-\nu}\;\ge\;\frac{1}{2\sqrt{\nu}}.
\label{eq:centralbinom}
\end{equation}

\emph{Central binomial bound.} We prove \Cref{eq:centralbinom} by induction
on $\nu$. At $\nu=1$ both sides equal $1/2$. Passing from $\nu$ to $\nu+1$
multiplies the left-hand side by $(2\nu+1)/(2\nu+2)$ and the right-hand side
by $\sqrt{\nu/(\nu+1)}$, so the inductive step holds provided the first ratio
is at least the second. Both are positive, so squaring and cross-multiplying,
this is $(2\nu+1)^2(\nu+1)\ge\nu(2\nu+2)^2$, that is,
$4\nu^3+8\nu^2+5\nu+1\ge 4\nu^3+8\nu^2+4\nu$, which reduces to $\nu+1\ge 0$.
Hence \Cref{eq:centralbinom} holds for all $\nu\ge1$, and
$\EE|S_T|\ge 2\nu\cdot\frac{1}{2\sqrt\nu}=\sqrt{\nu}=\sqrt{\lfloor T/2\rfloor}$.
\end{proof}

\section{Proofs for \texorpdfstring{\Cref{sec:block}}{Section 5}}
\label{app:block}

\subsection{The partial-information learner of \texorpdfstring{\citet{balcan2015}}{Balcan et al.}}
\label{app:balcan-alg}

The proof of \Cref{lem:anytime} needs more than the statement of
\Cref{thm:balcan61}: it reopens the argument behind it. This subsection
describes the learner $\Api(C)$ of \citet[Algorithm 1]{balcan2015} and the
two properties of it that the argument uses. Everything here is due to
\citet[Sec.~6]{balcan2015}; we restate it in our notation, with the catalog
$C$ in the role of their known type set, a horizon $H\le T$ known in
advance, and $\epsilon=1/T$ as in \Cref{alg:bgc}. Throughout, the attacker
sequence is oblivious and takes values in $C$, and we write
$u_t(\pvec)=U_d(b_{a_t}(\pvec),\pvec)\in[-1,1]$ for the payoff that $\pvec$
would earn at round $t$.

\paragraph{Experts and the obstacle.}
Like \Cref{alg:gc}, the learner competes only against the finite set
$\cE(C;\epsilon)$ of \Cref{lem:43}, which on any $H\le T$ rounds loses at
most $2\epsilon H\le 2$ against the best $\pvec\in\cP$. Under full
information one would run the expert algorithm of \Cref{prop:hedge} over
$\cE(C;\epsilon)$, charging every expert $\pvec$ the loss $-u_t(\pvec)$ at
every round. Under partial information the defender observes only the
attacked target $i_t$, and $u_t(\pvec)$ for $\pvec\ne\pvec_t$ depends on
$a_t$, so these losses are not available. \citet{balcan2015} instead give the
expert algorithm one \emph{estimated} loss vector per block of rounds. We
call it the \emph{background learner}.

\paragraph{Estimation blocks.}
The horizon $[1,H]$ is split into $Z$ consecutive \emph{estimation blocks}
$B_1,\dots,B_Z$ of equal length $G=H/Z$; these are unrelated to the
revelation blocks $\cB_\nu$. At the start of $B_\tau$ the background learner
holds a distribution $q_\tau$ over $\cE(C;\epsilon)$, with $q_1$ uniform.
Inside $B_\tau$, $d\le|C|$ rounds chosen uniformly at random are
\emph{exploration} rounds, on which fixed probing strategies, described
below, are played. On every other round, an \emph{exploitation} round, the
learner draws $\pvec_t\sim q_\tau$ and plays it. At the end of the block the
exploration rounds yield, for every $\pvec\in\cE(C;\epsilon)$, an estimate
$\hat c_\tau(\pvec)$ of the block's average payoff
\[
c_\tau(\pvec)\;=\;\frac{1}{G}\sum_{t\in B_\tau}u_t(\pvec),
\]
and the background learner receives the loss vector $-\hat c_\tau$ as its
$\tau$-th round and returns $q_{\tau+1}$. The number of blocks is
\citep[Algorithm 1]{balcan2015}
\begin{equation}
Z\;=\;n\bigl(H^{2}\log(n|C|)\bigr)^{1/3},
\label{eq:Z}
\end{equation}
rounded to an integer. It is chosen to balance the two error terms of
\Cref{eq:chain} below, and it is the only place where the learner uses $H$.

\paragraph{The estimates.}
Few probes suffice because, over a block, the payoff of every expert depends
on the attackers only through how often each type attacked. Let
$m_\tau\in\NN^{C}$ count the rounds of $B_\tau$ at which each type of $C$
attacked. Take $\pvec\in\cP_\sigma(C)$ (\Cref{eq:Psigma}), so that under
$\pvec$ each type $\alpha$ attacks $\sigma(\alpha)$, and let
$w_{\sigma,i}\in\{0,1\}^{C}$ mark the types $\alpha$ with $\sigma(\alpha)=i$.
Then $w_{\sigma,i}\cdot m_\tau$ is the number of rounds of $B_\tau$ at which
$\pvec$ would be answered by an attack on $i$, so
\begin{equation}
c_\tau(\pvec)\;=\;\frac{1}{G}\sum_{i\in N}U_d(i,\pvec)\;w_{\sigma,i}\cdot m_\tau .
\label{eq:ctau-app}
\end{equation}
A single probe estimates a single term $w_{\sigma,i}\cdot m_\tau/G$: play a
fixed $\pvec_w\in\cP_\sigma(C)$ at a uniformly random round of $B_\tau$ and
record $\hat p_\tau(w)=1$ if target $i$ is attacked and $0$ otherwise. Under
$\pvec_w$ the attacker strikes $i$ exactly when it is marked by
$w_{\sigma,i}$, and since the sequence is fixed in advance and the round is
uniform, $\EE[\hat p_\tau(w)]=w_{\sigma,i}\cdot m_\tau/G$
\citep[Lemma 6.3]{balcan2015}. The map $w\mapsto w\cdot m_\tau$ is linear,
so it suffices to probe a basis of the span of the finite set
$\mathcal{W}=\{w_{\sigma,i}\}$, whose dimension is $d\le|C|$. To keep the
estimates bounded, \citet{balcan2015} probe a \emph{barycentric spanner}
\citep{awerbuch2008}: a basis $b_1,\dots,b_d\in\mathcal{W}$ such that every
$w\in\mathcal{W}$ equals $\sum_{j}\lambda_j(w)\,b_j$ with all
$\lambda_j(w)\in[-1,1]$. Substituting into \Cref{eq:ctau-app} gives the
estimate \citep[Eq.~2]{balcan2015}
\begin{equation}
\hat c_\tau(\pvec)\;=\;\sum_{i\in N}U_d(i,\pvec)\sum_{j=1}^{d}
\lambda_j(w_{\sigma,i})\,\hat p_\tau(b_j),
\label{eq:chat-app}
\end{equation}
which has the two properties the argument uses
\citep[Lemma 6.4]{balcan2015}:
\begin{equation}
\EE[\hat c_\tau(\pvec)]=c_\tau(\pvec)
\qquad\text{and}\qquad
|\hat c_\tau(\pvec)|\le\kappa\coloneqq n|C| .
\label{eq:est-props}
\end{equation}
The first follows from linearity, the second from $|U_d|\le 1$,
$|\lambda_j|\le 1$, $\hat p_\tau\in\{0,1\}$ and $d\le|C|$. The exploration
rounds of $B_\tau$ are drawn afresh, independently of earlier blocks, so the
first property also holds conditionally on $q_\tau$.

\subsection{Proof of \texorpdfstring{\Cref{lem:anytime}}{Lemma}}
\label{app:anytime}

By \Cref{app:balcan-alg}, $\Api(C)$ uses its horizon $H$ only through the
number of estimation blocks $Z$ of \Cref{eq:Z}, chosen to balance
exploration against exploitation at that particular $H$. Inside an epoch the
horizon is the epoch length, which is unknown when the learner starts, so $Z$
cannot be set. This difficulty differs from the one addressed in
full information feedback: there the only horizon-dependent quantity was a learning
rate, which can be set without knowing the horizon. Here the horizon fixes
how often the algorithm explores, and no choice of background learner makes
that horizon-free, so we keep $\Api(C)$ and wrap it. Two ingredients
suffice: (i) the bound of \Cref{thm:balcan61}, though stated for a run tuned
to a horizon $H$, in fact bounds the regret at every prefix of that run, not
only at its end; (ii) a doubling wrapper removes the tuning.

\emph{(i) The horizon-tuned bound is also a prefix bound.}
Fix $H\le T$ and run $\Api(C)$ tuned to it. Fix $\tau_0\in\{1,\dots,Z\}$;
the first $\tau_0$ blocks are the rounds $t\le\tau_0G$. We follow the chain
of inequalities in the proof of \citet[Lemma 6.2]{balcan2015}, but stop it
after $\tau_0$ blocks instead of $Z$. It has three links.

\emph{Exploration.} On an exploitation round of $B_\tau$ the expected payoff
given $q_\tau$ is $q_\tau\cdot u_t$, where
$q_\tau\cdot u_t=\sum_{\pvec}q_\tau(\pvec)u_t(\pvec)$. The
$d$ exploration rounds earn some other payoff, and since all payoffs lie in
$[-1,1]$, replacing them by $q_\tau\cdot u_t$ changes the block's total by at
most $2d$. Hence the expected payoff on $B_\tau$ is at least
$G\,\EE[q_\tau\cdot c_\tau]-2d$.

\emph{Background learner.} By \Cref{eq:est-props}, conditionally on $q_\tau$,
$\EE[q_\tau\cdot\hat c_\tau]=q_\tau\cdot c_\tau$. After $\tau_0$ blocks the
background learner has received the loss vectors
$-\hat c_1,\dots,-\hat c_{\tau_0}$, with entries in $[-\kappa,\kappa]$, so by
\Cref{prop:hedge}, which holds at every round without knowing the number of
rounds in advance, it trails the best expert on the estimates by at most
$R(\tau_0)=O\bigl(\sqrt{\tau_0\,\kappa\log|\cE(C;\epsilon)|}\bigr)$.

\emph{Back to true payoffs.} The expectation of a maximum is at least the
maximum of the expectations, and by \Cref{eq:est-props} the latter is the
maximum of the true block averages.

Chaining the three links, multiplying by the block length $G$, and finally
passing from $\cE(C;\epsilon)$ to $\cP$ by \Cref{lem:43}:
\begin{align}
\EE\Bigl[\sum_{t\le\tau_0G}u_t(\pvec_t)\Bigr]
&\ge G\sum_{\tau\le\tau_0}\EE[q_\tau\cdot c_\tau]-2d\tau_0\notag\\
&= G\,\EE\Bigl[\sum_{\tau\le\tau_0}q_\tau\cdot\hat c_\tau\Bigr]-2d\tau_0\notag\\
&\ge G\,\EE\Bigl[\max_{\pvec\in\cE(C;\epsilon)}\sum_{\tau\le\tau_0}\hat c_\tau(\pvec)\Bigr]
   -G\,R(\tau_0)-2d\tau_0\notag\\
&\ge \max_{\pvec\in\cE(C;\epsilon)}\sum_{t\le\tau_0G}u_t(\pvec)
   -G\,R(\tau_0)-2d\tau_0\notag\\
&\ge \max_{\pvec\in\cP}\sum_{t\le\tau_0G}u_t(\pvec)
   -2-\underbrace{G\,R(\tau_0)}_{\text{background learner}}
   -\underbrace{2d\tau_0}_{\text{exploration}} .
\label{eq:chain}
\end{align}
So the regret over the first $\tau_0$ blocks is at most $2$ plus the two
error terms marked in \Cref{eq:chain}. Both are nondecreasing in $\tau_0$, so
for every $\tau_0\le Z$ they are at most their values at $\tau_0=Z$, which is
where \citet{balcan2015} evaluate the chain. There, with $d\le|C|$,
$\kappa=n|C|$, $\log|\cE(C;\epsilon)|=O(n^2|C|\log(n|C|))$ by \Cref{lem:44},
and $Z$ as in \Cref{eq:Z},
\[
G\,R(Z)=O\Bigl(H\sqrt{\kappa\log|\cE(C;\epsilon)|/Z}\Bigr)
\quad\text{and}\quad
2dZ\le 2|C|Z
\]
are both $O(H^{2/3}n|C|\log^{1/3}(n|C|))$: this is the bound of
\Cref{thm:balcan61} with horizon $H$ and type set $C$.

Two cases remain. A prefix $L\le H$ that ends inside a block adds at most
$G$ rounds after the last complete block, each contributing at most $2$ to
the regret, and $2G=2H^{1/3}/(n\log^{1/3}(n|C|))$ is of lower order. And if
$H$ is so small that a block cannot hold its $d$ exploration rounds, that is,
$G\le|C|$, then $H^{1/3}\le n|C|\log^{1/3}(n|C|)$, so the trivial bound $2H$
on the regret is already $O(H^{2/3}n|C|\log^{1/3}(n|C|))$. We conclude that
there is an absolute constant $c_0$ with
\begin{equation}
\text{regret at prefix }L\;\le\;c_0\,H^{2/3}\,n|C|\log^{1/3}(n|C|)
\label{eq:prefix-bound}
\end{equation}
for every $L\le H$. The right-hand side depends on $H$, not $L$: at short
prefixes we still pay for the horizon we tuned to. Converting $H^{2/3}$ into
$L^{2/3}$ is the task of part~(ii).

\emph{(ii) Doubling wrapper.}
Run $\Api(C)$ in \emph{stages} of exponentially growing length: stage
$s=0,1,2,\dots$ has length $2^{s}$ and is a fresh copy of $\Api(C)$ tuned to
horizon $2^{s}$; at the end of stage $s$ we discard its state and start stage
$s+1$. This variant needs no horizon in advance, and the sequence restricted
to any stage is still oblivious. This is the doubling trick of
\citet{auer1995} \citep[see also][Sec.~2.3]{cesa2006}, and the computation
below is that of
\citet[Theorem~4]{besson2018} at $\gamma=2/3$; what it requires from the base
learner is a bound on a stage that is cut off early, which part~(i)
provides. Fix any $L\le T$ and let $S$ be the stage
containing round $L$; stages $0,\dots,S-1$ occupy $2^{S}-1$ rounds, all
before $L$, so $2^{S}\le L$. Stage $S$ is in general cut off at $L$ before
its end, and this is why (i) is needed: \Cref{eq:prefix-bound} bounds the
regret of a run at every prefix, not only at the horizon it was tuned to.

Regret is subadditive across stages: for any $\pvec\in\cP$ the payoff over
$[1,L]$ decomposes as a sum of per-stage payoffs, and the maximum over $\cP$
of a sum is at most the sum of the per-stage maxima. Subtracting the
algorithm's payoff stage by stage,
\[
\text{regret over }[1,L]\;\le\;\sum_{s=0}^{S}\text{regret of stage }s
\;\le\;c_0\,\phi\sum_{s=0}^{S}2^{2s/3},
\]
where $\phi\coloneqq n|C|\log^{1/3}(n|C|)$ and each stage's regret comes
from \Cref{eq:prefix-bound} with $H=2^{s}$. The remaining sum is a
geometric series with ratio $2^{2/3}>1$, so summing it from the top
downwards,
\[
\sum_{s=0}^{S}2^{2s/3}\;\le\;\frac{2^{2/3}}{2^{2/3}-1}\cdot 2^{2S/3}
\;\le\;\frac{2^{2/3}}{2^{2/3}-1}\cdot L^{2/3},
\]
the last step using $2^{S}\le L$. Setting
$c_4\coloneqq 2^{2/3}c_0/(2^{2/3}-1)<2.71\,c_0$ yields
\Cref{eq:anytime}.\hfill$\square$

\subsection{Proof of \texorpdfstring{\Cref{thm:blocklb}}{Theorem}}
\label{app:blocklb}

The proof sketch after \Cref{thm:blocklb} rests on two devices.

\emph{Nested intervals, so that one strategy suits every type.} By
\Cref{eq:type-line}, a coverage vector $(p,0,1-p)$ steers the type with
interval $J$ precisely when $p\in J$. So if the adversary chooses $\Kmax$
types whose intervals are strictly nested, $J_1\supset J_2\supset\cdots\supset
J_{\Kmax}$, then any $p$ in the innermost interval lies in all of them and
therefore steers every type at once. The benchmark thus has value $0$ and,
as in \Cref{thm:impossible}, regret counts failed rounds.

\emph{Fresh randomness, so that each block is a new search.} The type
$\alpha^{[j]}$ attacks only during the $j$-th block, so the oracle names it
only at that block's closing checkpoint. Given $J_{j-1}$, the next interval
$J_j$ is a uniformly chosen one of the $M=2^B$ equal sub-intervals of
$J_{j-1}$. When block $j$ begins, the defender knows $J_{j-1}$ from the
previous checkpoint, but nothing about where inside $J_{j-1}$ the interval
$J_j$ lies. It must locate $J_j$ using the $B$ rounds of that block, which is
the search problem of \Cref{lem:count} with $M=2^B$ candidates and horizon
$B$.

\emph{Construction.}
Put $M=2^{B}$, draw $k_1,\dots,k_{\Kmax}$ independently and uniformly from
$[M]$, and define nested intervals $J_0=(0,1]\supset J_1\supset\cdots$ by
letting $J_j$ be the $k_j$-th of the $M$ equal sub-intervals of $J_{j-1}$:
\begin{equation}
J_j=\Bigl(\ell_{j-1}+\tfrac{k_j-1}{M}|J_{j-1}|,\ \
        \ell_{j-1}+\tfrac{k_j}{M}|J_{j-1}|\Bigr],
\label{eq:nested}
\end{equation}
$\ell_{j-1}$ and $|J_{j-1}|$ being the left endpoint and length of
$J_{j-1}$. Let $\alpha^{[j]}$ be the threshold type with interval $J_j$, put
$K=\{\alpha^{[1]},\dots,\alpha^{[\Kmax]}\}$, and let $a_t=\alpha^{[j]}$ for
$t\in\cB_j$ and $a_t=\alpha^{[\Kmax]}$ for $t>\Kmax B$. The sequence is fixed
once the $k_j$ are drawn, so the adversary is oblivious, and each $J_j$ is a
\emph{proper} sub-interval of $J_{j-1}$, so the $\Kmax$ types are distinct.
Fix an online algorithm and let $\overline{R}$ be its regret averaged over
the draw of $k_1,\dots,k_{\Kmax}$. As in the proof of \Cref{thm:impossible},
we first bound $\overline{R}$ for deterministic algorithms and then average
over the coins of a randomized one.

\emph{The benchmark is $0$.} Let $\mu$ be the midpoint of the innermost
interval $J_{\Kmax}$ and $\pvec^{*}=(\mu,0,1-\mu)$. Write
$J_j=(r_1^{[j]},r_2^{[j]}]$. Because the intervals are nested, each contains
$J_{\Kmax}$, so their left endpoints increase and their right endpoints
decrease with $j$. Hence $r_1^{[j]}\le r_1^{[\Kmax]}<\mu$ and
$1-r_2^{[j]}\le 1-r_2^{[\Kmax]}<1-\mu$ for every $j$, the strict
inequalities holding because $\mu$ lies strictly inside $J_{\Kmax}$. Both
conditions of \Cref{eq:hit-iff} are therefore met for every $j$, and
strictly, so $\pvec^{*}$ steers every type regardless of the tie-breaking
rule, and regret counts failed rounds.

\emph{Each block is a fresh search.} Let the algorithm be deterministic.
When $\cB_j$ begins, the checkpoint closing $\cB_{j-1}$ has revealed
$\alpha^{[1]},\dots,\alpha^{[j-1]}$, hence $J_{j-1}$, and nothing about
$\alpha^{[j]}$, which has not yet attacked. Moreover the whole transcript of
$\cB_1,\dots,\cB_{j-1}$ is a function of $k_1,\dots,k_{j-1}$ alone: this
follows by induction on the rounds, since each play is determined by earlier
observations, each observation by that play and the interval then attacking,
and the intervals active in those blocks by $k_1,\dots,k_{j-1}$. As $k_j$ is
drawn independently of $k_1,\dots,k_{j-1}$, it remains uniform on $[M]$
conditionally on everything the defender has seen and been told before
block $j$ starts.

Condition on that transcript. Within $\cB_j$ the defender's plays are
functions of the targets attacked inside $\cB_j$, the only attacker is
$\alpha^{[j]}$, and its interval is one of $M$ disjoint sub-intervals of
$J_{j-1}$. The oracle says nothing during $\cB_j$, so \Cref{lem:count}
applies with these $M$ candidates and horizon $B$, and gives
$\EE[\#\text{failures in }\cB_j]\ge B-(2^{B+1}-B-2)/2^{B}>B-2$.
The $\Kmax$ blocks are disjoint sets of rounds, so their failure counts
add, and summing over $j\in[\Kmax]$ gives $\overline{R}\ge\Kmax(B-2)$ for
every deterministic algorithm. A randomized algorithm is a distribution over
deterministic ones whose coins are independent of the draw, so averaging over
its coins gives the same bound. Finally, since $\overline{R}$ is an average
over the draw, some fixed draw of $k_1,\dots,k_{\Kmax}$ has regret at least
$\Kmax(B-2)$, and that fixed draw is the instance in the statement of
\Cref{thm:blocklb}.\hfill$\square$

\paragraph{Comparing the two bounds.}
\label{rem:blocktight}
\Cref{thm:block,thm:blocklb} bracket the cost of batching between
$\Kmax(B-2)$ and $2\Kmax B$, so it is $\Theta(\Kmax B)$ for $B\ge 3$; for
$B\le 2$ the lower bound is vacuous. The comparison needs one caveat:
\Cref{thm:blocklb} bounds the \emph{total} regret on its instance, whereas
the $2\Kmax B$ of \Cref{thm:block} is only the batching part of an upper
bound whose other part is the learning term. Take the random instance above
with $T=\Kmax B$, so that $|K'|=\Kmax$ and $T^{2/3}=\Kmax^{2/3}B^{2/3}$.
\Cref{thm:block} holds on every instance, hence also on average over the
draw, and together with the bound $\overline{R}\ge\Kmax(B-2)$ just proved it
gives
\[
\Kmax(B-2)\;\le\;\overline{R}\;\le\;
\underbrace{c\,n\Kmax^{2}B^{2/3}\log^{1/3}(n\Kmax)}_{\text{learning term}}
\;+\;2\Kmax B .
\]
Fix $\delta>0$. The learning term is at most $\delta\cdot 2\Kmax B$ exactly
when $B\ge (c/2\delta)^{3}\,n^{3}\Kmax^{3}\log(n\Kmax)$. For such $B$, and
$B\ge 3$, the upper bound is at most $(1+\delta)\,2\Kmax B$, so the upper and
lower bounds differ by a factor of at most
\[
\frac{(1+\delta)\,2\Kmax B}{\Kmax(B-2)}\;=\;(1+\delta)\Bigl(2+\frac{4}{B-2}\Bigr),
\]
which equals $6(1+\delta)$ at $B=3$ and $4(1+\delta)$ at $B=4$, and
decreases to $2(1+\delta)$ as $B\to\infty$. At the other extreme, $B=1$, the
additive term is $2\Kmax$ and \Cref{eq:block-main} becomes
$O(n\Kmax^{4/3}T^{2/3}\log^{1/3}(n\Kmax)+\Kmax)$.

\section{Experiments}
\label{sec:exp}

We test \Cref{alg:gc,alg:bgc} against the bounds of
\Cref{thm:main,thm:block}, and we play the game $G_3$ exactly to check that
the lower bounds of \Cref{thm:impossible,prop:sqrtT,thm:blocklb} are
attained.

\subsection{Setup}
\label{sec:exp-setup}

Besides \Cref{alg:gc}, we run the known-type learner of
\citet[Sec.~5]{balcan2015}, which is given $K$ before the first round, and a
\emph{uniform} baseline that plays $(R/n,\dots,R/n)$ in every round. Unless
stated otherwise, $n=3$, $R=1$, and the utilities are drawn independently
per target with fixed signs, $u_\alpha^c(i),u_d^u(i)\sim\mathrm{Unif}[-1,-0.1]$
and $u_\alpha^u(i),u_d^c(i)\sim\mathrm{Unif}[0.1,1]$,
so that coverage deters every attacker. At $n=3$ the expert set
$\cE(C;\epsilon)$ is computed by enumeration. At $n=10$ and $n=20$ this is
infeasible, and both algorithms instead use the optimization oracle of
\Cref{app:oracle}, which returns the best coverage vector for a given linear
objective without building $\cE(C;\epsilon)$.

We use two adversaries. The \emph{staggered-uniform} adversary splits the
horizon into $\Kmax$ phases of equal length and in phase $i$ draws the
attacker uniformly from the first $i$ types, so that new types arrive at
known times while old and new types keep mixing. The
\emph{uniform-random} adversary draws every attacker uniformly from $K$.
The staggered adversary is non-stationary, so an algorithm that tracks the
current phase can beat the best fixed coverage vector, and cumulative regret
against $\pvec^{*}$ can be negative (at $n=10$ it was on $14$ of the $40$ seeds). We
therefore use it only to show the epoch structure, and the stationary
uniform-random adversary to measure cumulative regret.

\Cref{alg:bgc} needs estimation blocks long enough to hold $|C|$ exploration
rounds. The block length $H/Z$ implied by \Cref{eq:Z} is asymptotic and
falls below one round at $n=10$, $|C|=4$, $T=2000$, so we use blocks of
$\max\{|C|+1,\lceil T^{1/3}\rceil\}$ rounds, which has the same order in
$T$. At $T=2000$ and $\Kmax=4$ that is $13$ rounds, of which $4$ explore.

At $n=3$ the inner learner is Hedge with rate
$\eta_t=\sqrt{\ln|\cE(C;\epsilon)|/t}$, and in \Cref{fig:block}
$B=\max\{4,\operatorname{round}(T^{2/3}/4)\}$ except in panel~(c).

\subsection{Full information}
\label{sec:exp-results}
\label{sec:exp-full}

\begin{figure}[!htbp]
\centering
\includegraphics[width=0.85\linewidth]{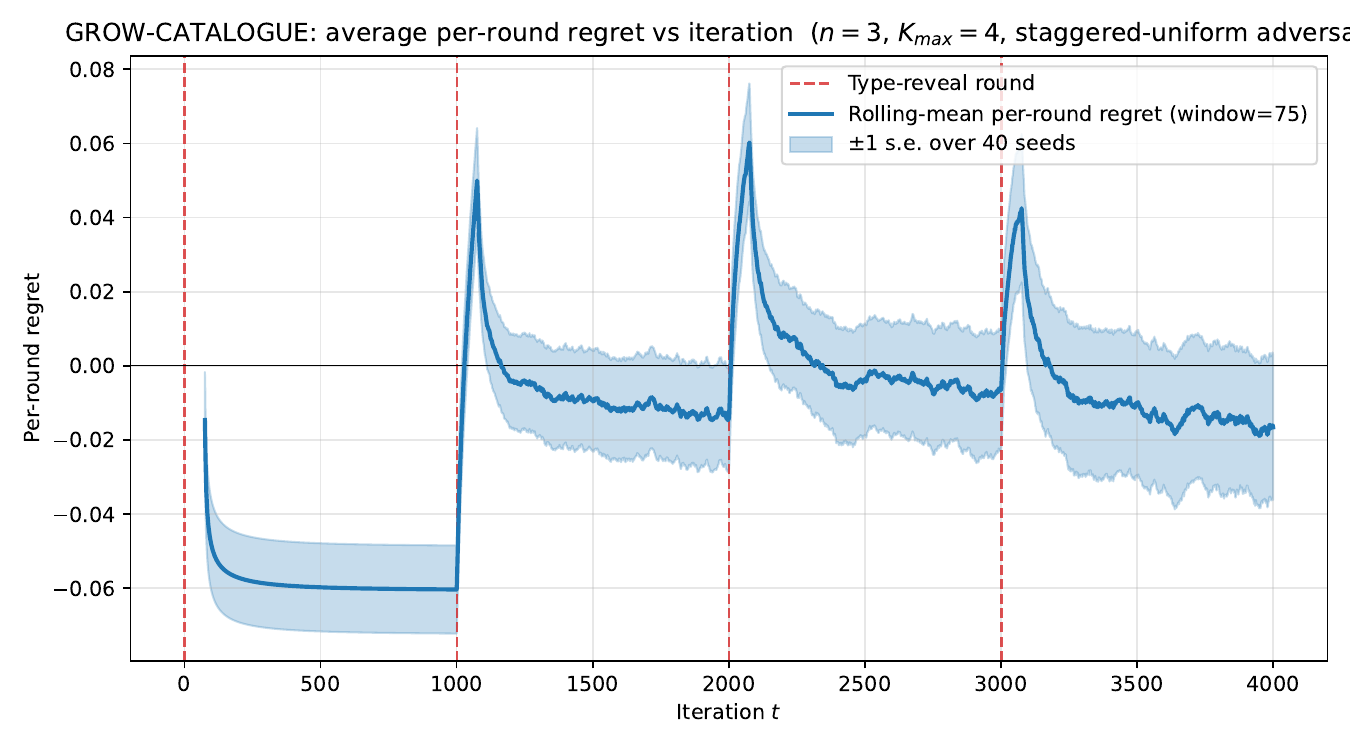}
\caption{Rolling-window per-round regret of \textsc{Grow-Catalog} under
the staggered-uniform adversary ($n=3$, $\Kmax=4$, $T=4000$, $40$ seeds,
window $75$). Dashed lines mark the discovery rounds $0,1000,2000,3000$.}
\label{fig:per-round}
\end{figure}

\begin{figure}[!htbp]
\centering
\includegraphics[width=0.85\linewidth]{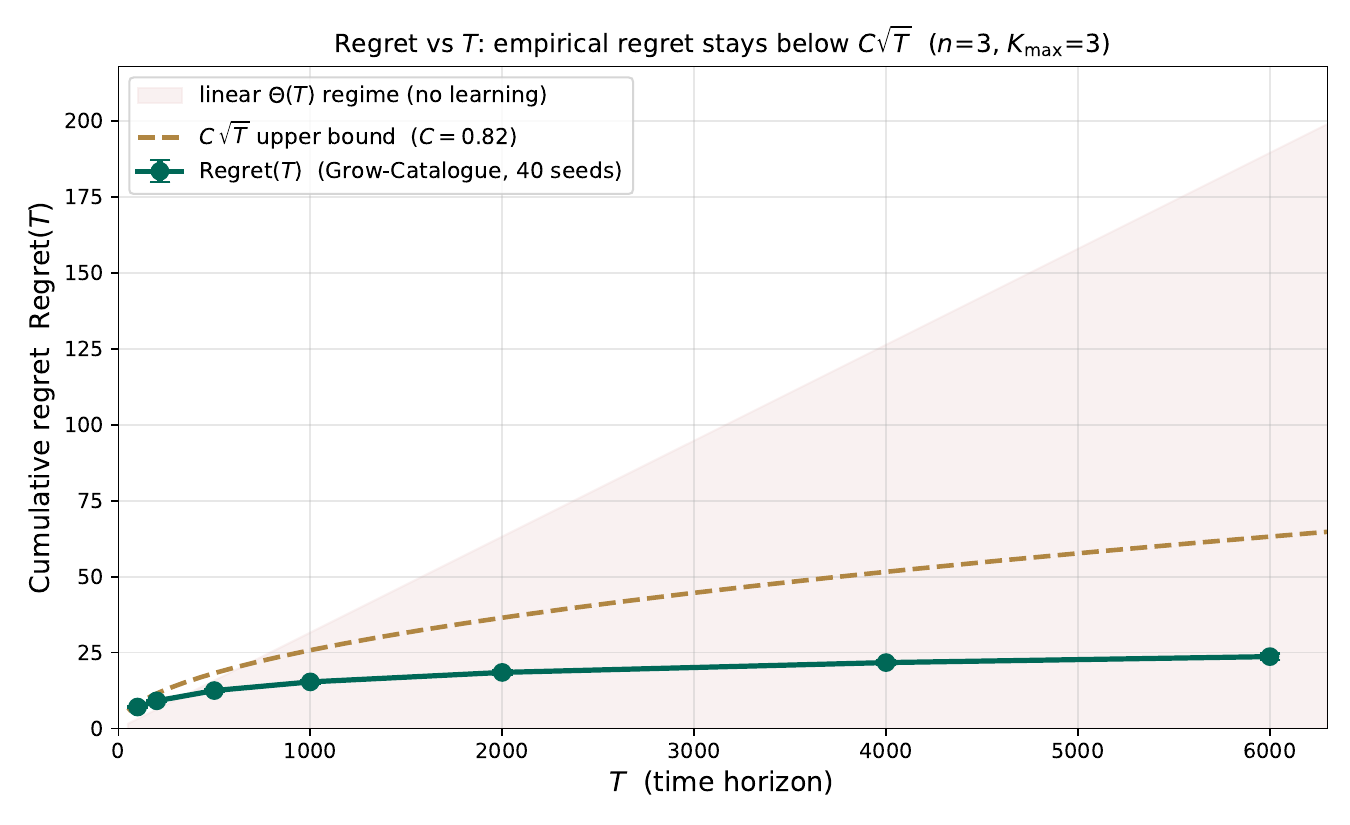}
\caption{Cumulative regret of \textsc{Grow-Catalog} against the horizon $T$
($n=3$, $\Kmax=3$, uniform-random adversary, $40$ seeds), with the
$C\sqrt T$ envelope of \Cref{thm:main}.}
\label{fig:scaling}
\end{figure}

\begin{table}[!htbp]
\centering
\footnotesize
\renewcommand{\arraystretch}{1.15}
\begin{tabular}{lcc}
\toprule
Algorithm & Regret & Ratio\\
\midrule
Known-type learner of \citet{balcan2015} & $21.95\ [16.67,\,24.84]$  & $1.00$\\
\textsc{Grow-Catalog}                     & $22.71\ [18.29,\,25.65]$  & $1.03$\\
Uniform baseline                          & $519.31\ [385.00,\,716.73]$ & $23.66$\\
\bottomrule
\end{tabular}
\caption{Regret at $T=4000$, $n=3$, $\Kmax=3$, uniform-random adversary,
median [interquartile range] over $40$ seeds.}
\label{tab:comparison}
\end{table}

\begin{figure}[!htbp]
\centering
\includegraphics[width=0.98\linewidth]{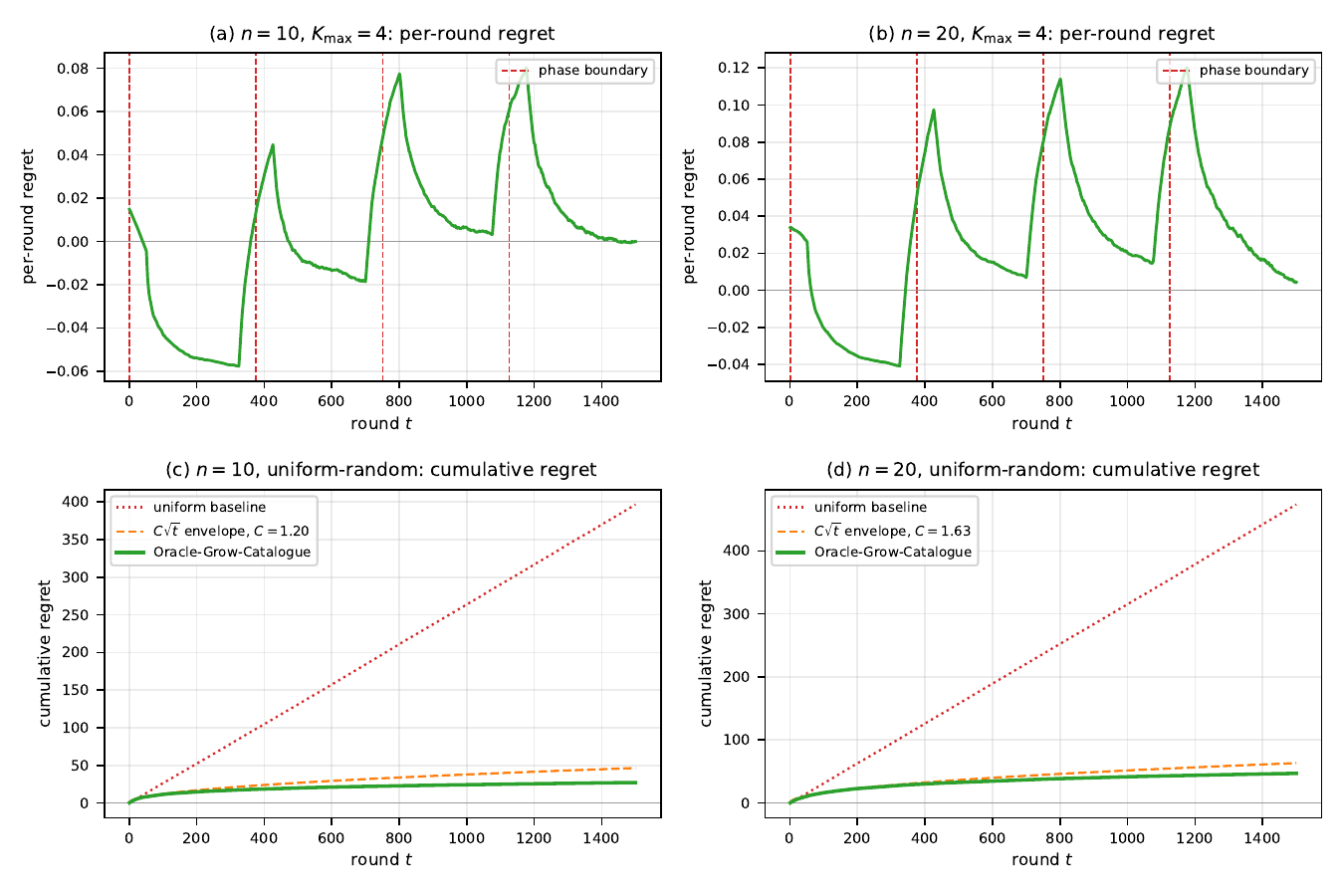}
\caption{\textsc{Grow-Catalog} with the oracle of \Cref{app:oracle},
$\Kmax=4$, $T=1500$, $40$ seeds. (a), (b) per-round regret under the
staggered-uniform adversary, rolling window $101$, dashed lines at the phase
boundaries. (c), (d) cumulative regret under the uniform-random adversary,
with the $C\sqrt t$ envelope of \Cref{thm:main} and the uniform baseline.}
\label{fig:tractable}
\end{figure}

\Cref{fig:per-round} shows the epoch structure of the proof of
\Cref{thm:main}: each discovery round produces a bounded spike in per-round
regret, which decays as the restarted learner converges within the new
epoch. \Cref{fig:scaling} shows the cumulative regret staying below the
$C\sqrt T$ envelope at every horizon tested, with $\Regret(T)/\sqrt T$
decreasing, so the stochastic adversary is far from the worst case.
\Cref{tab:comparison} shows that \textsc{Grow-Catalog} comes within $3\%$
of the learner that is given $K$ in advance, and is an order of magnitude
better than the uniform baseline.

\begin{figure}[!htbp]
\centering
\includegraphics[width=0.98\linewidth]{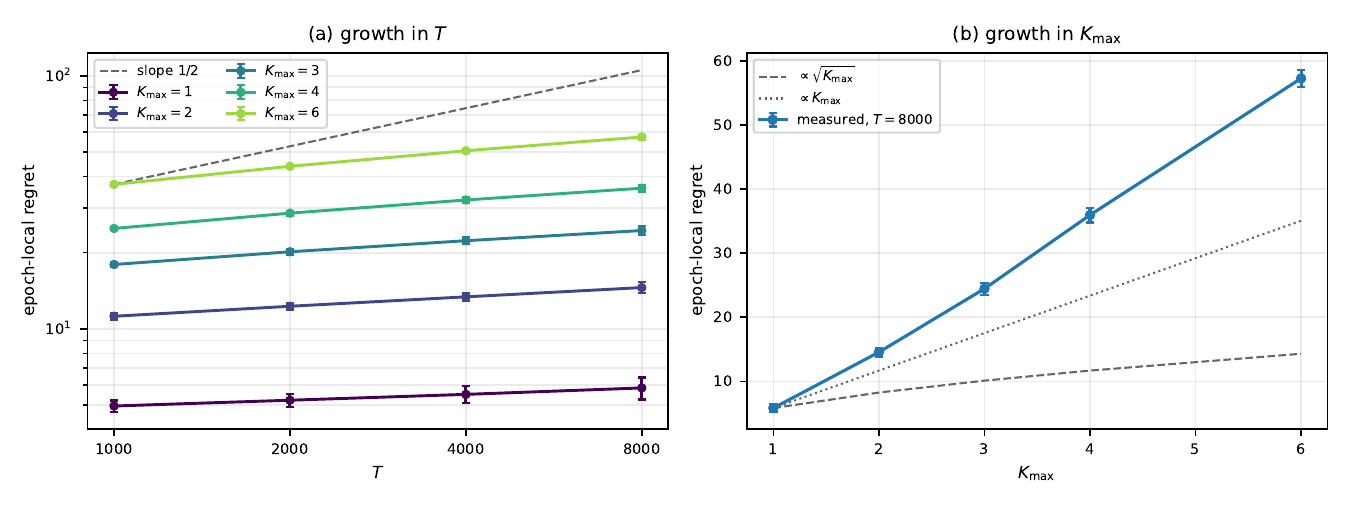}
\caption{The hard adversary: $\Kmax$ equal epochs, each attacked only by its
new type ($n=3$, $40$ seeds). (a) epoch-local regret against $T$, with a
slope-$1/2$ reference; (b) epoch-local regret against $\Kmax$ at $T=8000$.}
\label{fig:hard}
\end{figure}

\Cref{fig:hard} uses the adversary on which the proof of \Cref{thm:main} is
tight: $\Kmax$ equal epochs, each attacked only by its new type. Against the
best fixed strategy over the whole horizon the regret is negative ($-1032$ at
$\Kmax=4$, $T=8000$), since no single coverage vector suits every type, so we
report the sum over epochs of the regret against each epoch's own best
strategy, which is what the proof bounds. It grows far more slowly than
$\sqrt{\Kmax T}$: like $T^{0.08}$ to $T^{0.21}$, and slightly faster than
linearly in $\Kmax$, as each epoch restarts Hedge on a larger expert set,
reaching $57$ at $\Kmax=6$, $T=8000$, where $\sqrt{\Kmax T}\approx 219$. With
one type per epoch the best response never changes, so Hedge locks on within
a few rounds, and the bound is not attained.

\Cref{fig:tractable} repeats both experiments at $n=10$ and $n=20$. The
per-round regret again spikes after each phase boundary and decays within
the following epoch, and the cumulative regret stays below the $C\sqrt t$
envelope, with $C=1.20$ at $n=10$ and $C=1.64$ at $n=20$, while being $14.6$
and $10.1$ times smaller than that of the uniform baseline.

\subsection{Block revelation}
\label{sec:exp-block}
\label{app:bign}

\begin{figure}[!htbp]
\centering
\includegraphics[width=0.98\linewidth]{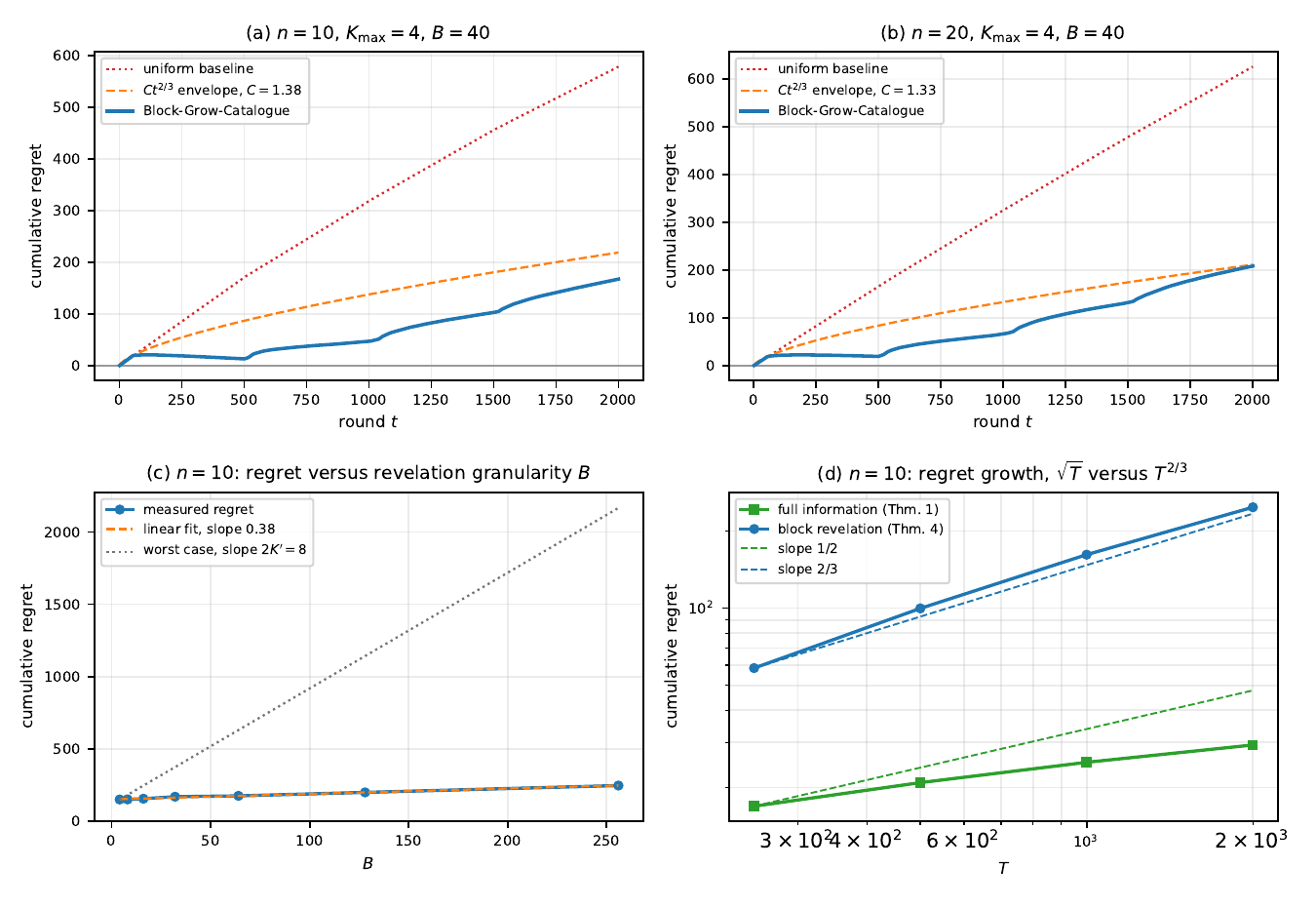}
\caption{\textsc{Block-Grow-Catalog} under partial information with block
revelation, $\Kmax=4$, $T=2000$, $40$ seeds. (a), (b) cumulative regret
under the staggered-uniform adversary against the $Ct^{2/3}$ envelope of
\Cref{thm:block}. (c) final regret against the revelation length $B$, with
a least-squares fit and the worst-case slope $2|K'|$. (d) full information
against block revelation as $T$ grows, uniform-random adversary.}
\label{fig:block}
\end{figure}

In \Cref{fig:block} the simulator passes \Cref{alg:bgc} only the attacked
target in each round, and at the end of each revelation block only the set
of types seen so far, which the algorithm reads through membership tests,
so no order or frequency information can leak. Panels (a) and (b) use the
staggered-uniform adversary, which spreads the discovery blocks over the
horizon and so produces the $\Theta(\Kmax)$ epochs that the proof of
\Cref{thm:block} sums over; the regret stays below the $Ct^{2/3}$
envelope, with $C=1.38$ at $n=10$ and $C=1.33$ at $n=20$. Panel (c) sweeps
$B$ over six doublings: the regret grows linearly in $B$, as
\Cref{thm:block,thm:blocklb} predict, with fitted slope $0.38$, well below
the worst case $2|K'|=8$, which only the nested instance of
\Cref{thm:blocklb} attains. Panel (d) compares the two feedback models on the
uniform-random adversary over $T\in\{250,500,1000,2000\}$: the fitted
exponents are $0.695$ ($95\%$ CI $[0.683,0.706]$, bootstrapped over seeds)
for block revelation, close to the $2/3$ of \Cref{thm:block}, and $0.264$
($[0.230,0.295]$) for full information, inside the $1/2$ of
\Cref{thm:main}. The difference has a simple source: \Cref{alg:bgc} spends
$|C|$ exploration rounds in every estimation block whatever the sequence,
about $31\%$ of all rounds at $T=2000$, whereas \Cref{alg:gc} has no such
floor.

\subsection{The lower bounds}
\label{sec:exp-imp}

\begin{figure}[!htbp]
\centering
\includegraphics[width=0.98\linewidth]{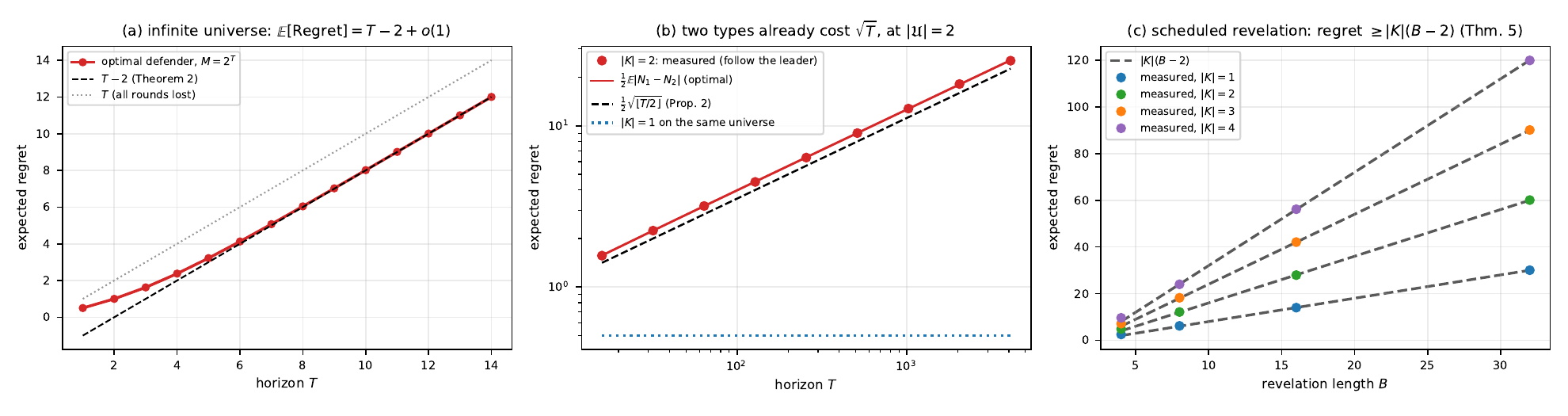}
\caption{The game $G_3$ in simulation, against the optimal defender.
(a) \Cref{thm:impossible}: with $|\Univ_M|=2^{T}$ the regret is $T-2+o(1)$.
(b) \Cref{prop:sqrtT}: two types cost $\Theta(\sqrt T)$, one type costs
$0.5$ at every horizon; logarithmic vertical axis.
(c) \Cref{thm:blocklb}: on the nested instance the regret matches
$|K|(B-2)$ (dashed lines).}
\label{fig:impossible}
\end{figure}

\Cref{fig:impossible} plays $G_3$ exactly: types are built from the
utilities \Cref{eq:type-utils}, best responses are computed from them, and
the expectation over the hidden type is computed by averaging over all $M$
candidates rather than by sampling. The defender is the midpoint search of
\Cref{rem:tight-lb}. In panel (a), with $M=2^T$, the regret equals
$T-2+(T+2)2^{-T}$ to machine precision for every $T\le 14$, confirming that
\Cref{thm:impossible} is exactly attained. In panel (b), with two types,
follow the leader has regret $1.57$, $6.35$ and $25.60$ at $T=16$, $256$ and
$4096$, against the exact value $\tfrac12\EE|N_1-N_2|=1.571$, $6.377$,
$25.531$ and the bound $\tfrac12\sqrt{\lfloor T/2\rfloor}=1.41$, $5.66$,
$22.63$ of \Cref{prop:sqrtT}; with one type the regret is $0.5$ at every
horizon. In panel (c), on the nested instance, the regret tracks $|K|(B-2)$
for $|K|\in\{1,2,3,4\}$ and $B\in\{4,8,16,32\}$, so the additive term of
\Cref{thm:block} is real.

\section{An Oracle-Efficient Implementation}
\label{app:oracle}

The experiments at $n=10$ and $n=20$ cannot build the expert set
$\cE(C;\epsilon)$, whose construction already takes minutes at $n=5$
(\Cref{tab:oracle}). This appendix describes how they avoid it. The idea is
that both algorithms only ever need a linear optimum over the experts, and
such an optimum can be computed by one Bayesian-Stackelberg optimization.

\paragraph{The lift.}
For a catalog $C$ and $\pvec\in\cP$ define
\begin{equation}
y_{\alpha,i}(\pvec)=\mathbb{I}\{b_\alpha(\pvec)=i\},\qquad
w_{\alpha,i}(\pvec)=y_{\alpha,i}(\pvec)\,p_i
\label{eq:lift}
\end{equation}
for $(\alpha,i)\in C\times N$, and let
$\mathcal{S}(C)=\{(y(\pvec),w(\pvec)):\pvec\in\cP\}\subseteq\RR^{2|C|n}$.

\begin{lemma}[both algorithms are linear over the lift]
\label{lem:lift}
Assume the types in $C$ have pairwise distinct best-response maps, so that
the set $\mathcal{W}$ of \Cref{app:balcan-alg} has rank $|C|$ and the matrix
$\Lambda$ whose rows are its barycentric spanner $b_1,\dots,b_{|C|}$ is
invertible. For any weights $\omega\in\RR^{C}$,
\begin{equation}
\sum_{\alpha\in C}\omega_\alpha\,U_d(b_\alpha(\pvec),\pvec)
=\sum_{\alpha\in C}\sum_{i\in N}\Bigl[\omega_\alpha u_d^u(i)\,y_{\alpha,i}(\pvec)
+\omega_\alpha\bigl(u_d^c(i)-u_d^u(i)\bigr)w_{\alpha,i}(\pvec)\Bigr],
\label{eq:lift-linear}
\end{equation}
which is linear in the lift. Moreover:
\begin{enumerate}[leftmargin=1.6em,itemsep=0.05em,topsep=0.2em,label=(\roman*)]
  \item the loss used by \Cref{alg:gc} at round $t$ is
  \Cref{eq:lift-linear} with $\omega=-\mathbb{I}_{\{a_t\}}$;
  \item the estimate $\hat c_\tau$ of \Cref{eq:chat-app}, used inside
  \Cref{alg:bgc}, is \Cref{eq:lift-linear} with $\omega=g_\tau$, where
  $g_\tau=\Lambda^{-1}\hat p_\tau$ and $\hat p_\tau\in\{0,1\}^{|C|}$
  collects the block's exploration outcomes; $g_\tau$ is an unbiased
  estimate of the per-round type frequencies in the block.
\end{enumerate}
\end{lemma}

\begin{proof}
For each $\alpha$ exactly one of $y_{\alpha,1},\dots,y_{\alpha,n}$ equals
$1$, so $U_d(b_\alpha(\pvec),\pvec)=\sum_i y_{\alpha,i}(\pvec)U_d(i,\pvec)$,
and by \Cref{eq:Ud}
$y_{\alpha,i}U_d(i,\pvec)=u_d^u(i)y_{\alpha,i}+\bigl(u_d^c(i)-u_d^u(i)\bigr)w_{\alpha,i}$.
Summing against $\omega$ gives \Cref{eq:lift-linear}, and (i) is the case
$\omega=-\mathbb{I}_{\{a_t\}}$. For (ii), let $\lambda(w)$ be the unique
vector with $\lambda(w)^{\top}\Lambda=w^{\top}$. Then
$\sum_j\lambda_j(w)\hat p_\tau(b_j)=w^{\top}\Lambda^{-1}\hat p_\tau=w\cdot g_\tau$,
and the $\alpha$-th coordinate of $w_{\sigma,i}$ is $y_{\alpha,i}(\pvec)$, so
\Cref{eq:chat-app} becomes
$\hat c_\tau(\pvec)=\sum_i\bigl(\sum_\alpha g_{\tau,\alpha}y_{\alpha,i}(\pvec)\bigr)U_d(i,\pvec)
=\sum_\alpha g_{\tau,\alpha}U_d(b_\alpha(\pvec),\pvec)$. Finally
$\EE[\hat p_\tau]=\Lambda m_\tau/|B_\tau|$ by \citet[Lemma 6.3]{balcan2015},
so $\EE[g_\tau]=m_\tau/|B_\tau|$.
\end{proof}

\paragraph{The oracle and the learner.}
By \Cref{lem:lift}, both algorithms are instances of online linear
optimization over $\mathcal{S}(C)$, whose only primitive is
\begin{equation}
\textsc{Oracle}(A,\Delta)=\argmax_{\pvec\in\cP}\ \sum_{\alpha,i}
\bigl(A_{\alpha,i}\,y_{\alpha,i}(\pvec)+\Delta_{\alpha,i}\,w_{\alpha,i}(\pvec)\bigr).
\label{eq:oracle}
\end{equation}
The implementation therefore replaces the Polynomial Weights learner of the
analysis by Follow the Perturbed Leader over $\mathcal{S}(C)$
\citep{kalai2005}, which calls \Cref{eq:oracle} once per round and not at all
on the exploration rounds of \Cref{alg:bgc}. This changes the learner but
not the structure of either algorithm: the catalog, the restarts and the
estimation blocks are as analyzed. The guarantee is of the same order in
$T$. Follow the Perturbed Leader has regret $O(\sqrt{DRAT})$ for linear
losses, where $D$ is the $\ell_1$-diameter of the decision set, $R$ bounds
the per-round loss and $A$ the $\ell_1$-norm of a loss vector
\citep{kalai2005}. On $\mathcal{S}(C)$ these are $D\le 4|C|$, $R\le 1$ and
$A\le 3n$, the last because
$\sum_i(|u_d^u(i)|+|u_d^c(i)-u_d^u(i)|)\le 3n$. We run its anytime
perturbation schedule, since segment lengths are not known in advance. The
constants of \Cref{thm:main,thm:block} are those of Polynomial Weights, and
the experiments test the exponents, not the constants.

\paragraph{Computing the oracle.}
Grouping \Cref{eq:oracle} by the response profile $\sigma:C\to N$ makes the
objective linear on each region $\cP_\sigma(C)$ of \Cref{eq:Psigma}, so
\Cref{eq:oracle} is the maximum of $n^{|C|}$ linear programs in $n$
variables, the extension to a set of types of the multiple-LP method of
\citet{conitzer2006}. In the experiments we instead solve one mixed-integer
program with $n|C|$ binary variables, in the style of the
Bayesian-Stackelberg program of \citet{paruchuri2008}, and use the profile
enumeration as a correctness check. We do not claim polynomial time in $|C|$,
since the Bayesian Stackelberg problem is NP-hard \citep{conitzer2006}; the
gain is in $n$, the parameter that made $\cE(C;\epsilon)$ unusable. Both
methods solve over regions shrunk by a margin $\gamma=10^{-6}$, so that the
returned $\pvec$ realizes its profile under any tie-breaking rule; this plays
the role of the $\epsilon$-approximation in \Cref{lem:43}. As in the
experiments, $\cP=\{\pvec\ge 0:\sum_i p_i=1\}$ with one resource, and the
algorithm and the benchmark optimize over the same set.

\paragraph{The rank condition.}
The rank hypothesis of \Cref{lem:lift} is a normalization of the catalog,
not an assumption on the game. Two types with the same best-response map on
all of $\cP$ are indistinguishable to the defender's payoff, to the benchmark
\Cref{eq:pstar} and to the feedback, so the catalog may keep one
representative of each map. This can only shrink the catalog, so every bound
of \Cref{sec:block} is unaffected. The implementation does not remove
duplicates but detects them: if no full-rank spanner exists, it records a
rank failure and skips exploration until the catalog next changes. No rank
failure occurred in the runs behind \Cref{fig:block}.

\paragraph{The spanner and the range of the estimates.}
The implementation does not enumerate $\mathcal{W}$. It samples coverage
vectors, records the vectors $w_{\sigma,i}$ they realize, and runs the
greedy determinant maximization of \citet{awerbuch2008} on that pool. A
spanner of a sub-pool need not give coefficients $\lambda\in[-1,1]$ for every
$w\in\mathcal{W}$, so the range bound $|\hat c_\tau|\le n|C|$ of
\Cref{eq:est-props} is not guaranteed. Unbiasedness, which is what the
estimator's correctness rests on, needs only that $\Lambda$ be invertible,
and holds exactly. For the same reason the implementation scales the
perturbation by the range of the estimates it actually computes rather than
by the worst case $n|C|$. This is a tuning choice: \Cref{thm:block} does not
rest on these constants, and \Cref{fig:block} tests the exponent.

\paragraph{Verification.}
On $60$ random instances with $n\in\{3,4,5\}$ and $|C|\in\{1,2,3\}$, the two
oracle methods and a brute-force maximization over the enumerated
$\cE(C;\epsilon)$ agree to within $1.4\times10^{-6}$. Over $4\times10^{4}$
replications of one estimation block at $n=6$, $|C|=3$, the mean of
$|B_\tau|\,g_\tau$ matches the true type counts of the block to within $0.05$
out of $40$, inside the Monte Carlo error of $1.1$ (three standard errors),
as \Cref{lem:lift}(ii) predicts.

\paragraph{Tractability.}
\label{sec:exp-tract}
\Cref{tab:oracle} times the three routines. Building $\cE(C;\epsilon)$
takes $0.24$ seconds at $n=3$, $|C|=4$, already $20$ minutes at $n=5$, $|C|=2$,
and is out of reach beyond that, whereas one oracle call costs $145$\,ms at
$n=10$, $0.9$\,s at $n=20$ and $1.6$\,s at $n=40$. Since neither algorithm makes more than one oracle call
per round, this is what moves the experiments from $n=3$ to $n=20$.

\begin{table}[t]
\centering
\footnotesize
\renewcommand{\arraystretch}{1.15}
\begin{tabular}{rrrrrr}
\toprule
$n$ & $|C|$ & $\binom{H}{n}$ & build $\cE$ & $n^{|C|}$ LPs & MILP\\
\midrule
 3 & 4 & $3.3\!\times\!10^{3}$ & $242$\,ms ($|\cE|=444$) & $30$\,ms & $19$\,ms\\
 5 & 4 & $3.5\!\times\!10^{7}$ & infeasible & $209$\,ms & $64$\,ms\\
10 & 4 & $1.2\!\times\!10^{19}$ & infeasible & $3.2$\,s & $145$\,ms\\
20 & 4 & $2.1\!\times\!10^{45}$ & infeasible & $65$\,s & $0.93$\,s\\
30 & 4 & $7.7\!\times\!10^{73}$ & infeasible & infeasible & $0.89$\,s\\
40 & 4 & $9.0\!\times\!10^{103}$ & infeasible & infeasible & $1.60$\,s\\
\bottomrule
\end{tabular}
\caption{Cost of one step, with $H=n+1+2|C|\binom{n}{2}$ the number of
hyperplanes. Column 3 is the number of $n$-subsets that the exhaustive
construction of $\cE(C;\epsilon)$ must solve. Columns 5 and 6 are the two
methods for \Cref{eq:oracle}; wherever both ran they returned the same
optimum to within $2\times 10^{-6}$. Median of $10$ calls. The MILP time is
not monotone in $n$ because it is an exact branch-and-bound solve.}
\label{tab:oracle}
\end{table}

\end{document}